\documentclass[11pt,english]{article}
\usepackage{mathptmx}
\usepackage[T1]{fontenc}
\usepackage[latin9]{inputenc}
\usepackage{float}
\usepackage{mathtools}
\usepackage{enumitem}
\usepackage{amsmath}
\usepackage{amsthm}
\usepackage{amssymb}
\usepackage[a4paper]{geometry}
\usepackage{setspace}
\usepackage[authoryear]{natbib}
\makeatletter
\newlength{\lyxlabelwidth}      % auxiliary length 
	\newenvironment{elabeling}[2][]%
	{\settowidth{\lyxlabelwidth}{#2}
		\begin{description}[font=\normalfont,style=sameline,
			leftmargin=\lyxlabelwidth,#1]}
	{\end{description}}
\theoremstyle{remark}
\newtheorem*{rem*}{\protect\remarkname}
\theoremstyle{plain}
\newtheorem{lem}{\protect\lemmaname}
\theoremstyle{definition}
\newtheorem{defn}{\protect\definitionname}
\theoremstyle{plain}
\newtheorem{thm}{\protect\theoremname}
\theoremstyle{plain}
\newtheorem{prop}{\protect\propositionname}
\theoremstyle{plain}
\newtheorem{cor}{\protect\corollaryname}

\DeclareMathOperator{\Var}{\mathrm{Var}}

\usepackage{algorithm,algpseudocode}

\usepackage{caption}

\numberwithin{equation}{section}

\usepackage{enumitem}   
\setlist[itemize]{noitemsep,topsep=2pt}
\setlist{noitemsep,topsep=2pt} 

\usepackage[font=small,labelfont=bf]{caption}

\usepackage{enumitem}
\setlist[enumerate,1]{label=(\alph*),font=\normalfont}

\makeatother

\usepackage{babel}
\providecommand{\corollaryname}{Corollary}
\providecommand{\definitionname}{Definition}
\providecommand{\lemmaname}{Lemma}
\providecommand{\propositionname}{Proposition}
\providecommand{\remarkname}{Remark}
\providecommand{\theoremname}{Theorem}

\begin{document}
\title{Mean-Variance Portfolio Selection in Long-Term Investments with Unknown
Distribution: Online Estimation, Risk Aversion under Ambiguity, and
Universality of Algorithms}
\author{Duy Khanh Lam\thanks{I am grateful to my advisor, Giulio Bottazzi, and Daniele Giachini
at Scuola Superiore Sant'Anna for their discussions, comments, and
encouragement, which motivated the completion of this paper.}\\Scuola Normale Superiore\\ \\(Working paper, 2nd full draft)}
\maketitle
\begin{abstract}
The standard approach for constructing a Mean-Variance portfolio involves
estimating parameters for the model using collected samples. However,
since the distribution of future data may not resemble that of the
training set, the out-of-sample performance of the estimated portfolio
is worse than one derived with true parameters, which has prompted
several innovations for better estimation. Instead of treating the
data without a timing aspect as in the common training-backtest approach,
this paper adopts a perspective where data gradually and continuously
reveal over time. The original model is recast into an online learning
framework, which is free from any statistical assumptions, to propose
a dynamic strategy of sequential portfolios such that its empirical
utility, Sharpe ratio, and growth rate asymptotically achieve those
of the true portfolio, derived with perfect knowledge of the future
data.\smallskip

When the distribution of future data follows a normal shape, the growth
rate of wealth is shown to increase by lifting the portfolio along
the efficient frontier through the calibration of risk aversion. Since
risk aversion cannot be appropriately predetermined, another proposed
algorithm updates this coefficient over time, forming a dynamic strategy
that approaches the optimal empirical Sharpe ratio or growth rate
associated with the true coefficient. The performance of these proposed
strategies can be universally guaranteed under stationary stochastic
markets. Furthermore, in certain time-reversible stochastic markets,
the so-called Bayesian strategy utilizing true conditional distributions,
based on past market information during investment, does not perform
better than the proposed strategies in terms of empirical utility,
Sharpe ratio, or growth rate, which, in contrast, do not rely on conditional
distributions.\bigskip
\end{abstract}
\begin{center}
{\small\textit{Keywords}}{\small : Online Learning, Misspecification,
Universality, Mean-Variance Portfolio, Dynamic Strategy.}{\small\par}
\par\end{center}

%\vspace{1\baselineskip}

\pagebreak{}

\setlength{\abovedisplayskip}{2.5pt} 
\setlength{\abovedisplayshortskip}{2.5pt}
\setlength{\belowdisplayskip}{2.5pt} 
\setlength{\belowdisplayshortskip}{2.5pt} 

\section{Introduction}

In this paper, we investigate the problem of sequential portfolio
selection for long-term investment under an unknown distribution of
future market data. An investor (or a fund manager) utilizes the Mean-Variance
(M-V) model to construct a no-short constrained portfolio considering
personal risk aversion, then frequently rebalances the portfolios
at each time period to maintain the risk profile aligned with the
original M-V portfolio throughout the investment. An advantage of
such constant rebalancing, as opposed to a varying sequence of portfolios,
is that it allows the investor to control the risk target of the investment
much better. However, as future data are unforeseeable, the investor
has to use the collected historical data to estimate model parameters
corresponding to the sample distribution and expect the distribution
of out-of-sample data to be similar. Technically, in accordance with
the common assumption of an identically and independently distributed
(i.i.d.) market, since the size of sample data is often not large
and costly to collect, some resampling methods on the sample are typically
employed to imitate an infinite sample so that the estimations could
be close to the true parameters of the hypothetical distribution.
Unfortunately, this is not often the case in reality as the future
data are gradually and constantly generated, and their distribution
is unlikely to be similar to the sample's distribution. Furthermore,
backtest overfitting resulting from training-testing dataset selection
could degrade empirical performance, as discussed in \citet{Bailey2014,Bailey2021a,Bailey2021}.\smallskip

The M-V model, introduced by \citet*{Markowitz1952,Markowitz1959},
has been the cornerstone of modern portfolio theory, owing to its
capability of controlling risk; for instance, it is one of the primary
models among large insurance companies, as referred to in the survey
by \citet{Grundl2016}. Roughly speaking, given the investor's risk
aversion $\alpha$, which measures the relative preference of the
expected return against the risk of a portfolio, the M-V portfolio
is a vector of weights of capital allocated to some assets, denoted
by $b^{*}$, that maximizes the tradeoff between the expected value
and variance of returns $r\big(b\big)$ of a portfolio $b$, which
is perceived as the inherent risk, as:
\[
b^{*}=\operatorname*{argmax}_{b}\big(\mathbb{E}\big(r\big(b\big)\big)-\alpha\Var\big(r\big(b\big)\big),\alpha\geq0.\vspace{-0.5ex}
\]
The values $\mathbb{E}\big(r\big(b\big)\big)$ and $\Var\big(r\big(b\big)\big)$
are computed through the expected value $\mu$ and covariance matrix
$\Sigma$ of the assets' returns that are assumed to be known in the
original theory. In reality, $\mu$ and $\Sigma$ are unknown and
commonly estimated using collected historical data, which are seriously
prone to error and coined as the term \textit{``Markowitz optimization
enigma''} by \citet*{Michaud1989}. As a consequence, the misestimation
results in inadequate out-of-sample performance of the estimated M-V
portfolio, especially its empirical Sharpe ratio, as shown in several
articles such as \citet*{Elton2014}, \citet*{Chopra1993}, \citet{Simaan1997},
\citet{BrittenJones1999} , \citet*{Frankfurter1976}, \citet*{Jobson1980}.
Thus, the \textit{estimated efficient frontier}, as termed by \citet{Broadie1993}
to refer to the frontier corresponding to the estimators, is not reliable
for calibrating a risk aversion coefficient.\smallskip 

A natural solution that emerged is improving the estimation for the
M-V model, which has become an intriguing statistical problem in recent
decades. Several approaches have been proposed, such as Bayesian estimators,
which assume a prior distribution for parameters of i.i.d. normal
assets' returns, and shrinkage methods along with their variants,
as in \citet*{Jorion1986}, \citet*{Black1991}, \citet*{Ledoit2003,Ledoit2004},
\citet{Fabozzi2007}, \citet{Lai2011} and \citet{Meucci2010}, among
others, and bootstrap resampling methods in \citet{Michaud1989}.
Other alternatives include using factor models to reduce dimensions
in covariance matrix estimation, such as the capital asset pricing
model by \citet{Sharpe1964}, or shrinkage by combining with a naive
portfolio in an i.i.d. market to mitigate estimation error as discussed
in \citet{Tu2011} and \citet{Lehmann1998}, etc. Additionally, related
to the indetermination of risk aversion, an alternative concept of\textit{
ambiguity }and\textit{ ambiguity aversion} under an unknown distribution
was proposed in \citet*{Garlappi2007} and \citet{Boyle2012}. Nonetheless,
the out-of-sample performance of all these methods in long-term investment
still remains a question.\smallskip

\textbf{Paper organization and main results}. This paper contributes
a novel solution to the described challenge by treating the future
data as gradually and constantly observed over periods of the investment,
rather than as a whole available at once for training (estimation)
and testing (backtesting). In this so-called online learning framework,
at any point in the future, the investor observes a set of data and
perceives its empirical distribution, allowing for the determination
of an M-V portfolio with respect to the true model parameters. Therefore,
this online approach does not require knowledge or assumptions of
the true distribution. The organization of this paper along with the
respective obtained results is outlined as follows:\smallskip
\begin{elabeling}{00.00.0000}
\item [{\textit{Section~2}:}] Taking into account the changing and unknown
future data, an online framework is formalized, in which a dynamic
strategy of sequential portfolios, rather than an estimated M-V portfolio
using a single sample, should be made such that its empirical M-V
utility, Sharpe ratio, and growth rate could asymptotically approach
those of the true constant M-V strategy, which is a portfolio derived
from the M-V model using perfect knowledge of the distribution of
the whole future data.$\vspace{0.5ex}$
\item [{\textit{Section~3.1}:}] In a market without redundant assets,
an algorithm is proposed for dynamic strategy construction with a
predetermined risk aversion, which guarantees the same empirical M-V
utility, Sharpe ratio, and growth rate as the constant M-V strategy
derived using the true limiting distribution of any deterministic
sequences of market data.$\vspace{0.5ex}$
\item [{\textit{Section~3.2}:}] If the market data exhibits the shape
of a normal distribution over time, a positive relation between the
growth rate and the change in both expected return and Sharpe ratio
of the constant M-V strategy is established. This explains the increase
in growth rate and the cumulative wealth of the M-V strategy along
the efficient frontier with lower calibration of the risk aversion
coefficient.$\vspace{0.5ex}$
\item [{\textit{Section~3.3}:}] Since the proposed algorithm in Section
3.1 requires a predetermined risk aversion, which is ambiguous due
to the unknown distribution of future data, a second algorithm is
proposed on top of the original one for dynamic strategy construction
with updating risk aversion coefficients. This allows the strategy
to approach either the highest empirical Sharpe ratio or the highest
growth rate among a set of constant M-V strategies.$\vspace{0.5ex}$
\item [{\textit{Sections~4}:}] The consistency of the proposed algorithms
can be guaranteed in the stationary stochastic markets. Moreover,
under a particular time-reversible market process, the so-called Bayesian
strategy using the M-V portfolios derived with the true conditional
distributions, based on past market information at each period, yields
the asymptotic empirical M-V utility, Sharpe ratio, or growth rate
almost surely not higher than those of the proposed strategies, which
do not require knowing the conditional distributions.
\end{elabeling}

\section{Model settings and problem formalization}

Let us consider the investment in the stock market with $m\geq2$
risky assets over discrete time periods $n\in\mathbb{N}_{+}$. For
clarity, let the vector $p_{n}=\big(p_{n,1},...,p_{n,m}\big)$ denote
the prices of the $m$ assets at time $n$, where $p_{n,j}\in\mathbb{R}_{++}$
is the price of the $j$-th asset. Correspondingly, the real-valued
random vector $X_{n}\coloneqq\big(X_{n,1},...,X_{n,m}\big)\in\mathcal{M}$
represents the positive returns of the assets at time $n$, where
$\mathcal{M}\coloneqq\big\{\beta\in\mathbb{R}_{++}^{m}:\,\|\beta\|\leq M\big\}$
for a real value $M\gg0$, and the vector $x_{n}\coloneqq\big(x_{n,1},...,x_{n,m}\big)$
indicates its corresponding realizations, where $x_{n,j}\coloneqq p_{n,j}/p_{n-1,j}$
for the $j$-th asset. At a specific time $n+1$, let $x_{1}^{n}\coloneqq\left\{ x_{i}\right\} _{i=1}^{n}$
be shorthand for the sequence of past realizations of the sequence
$X_{1}^{n}\coloneqq\left\{ X_{i}\right\} _{i=1}^{n}$; thus, the infinite
sequence of random vectors and their corresponding realizations are
denoted by $X_{1}^{\infty}$ and $x_{1}^{\infty}$, respectively.
The space of no-short portfolios is denoted by the simplex $\mathcal{B}^{m}\coloneqq\big\{\beta\in\mathbb{R}_{+}^{m}:\,{\displaystyle {\textstyle \sum}}_{{\scriptstyle j=1}}^{{\scriptstyle m}}\beta_{j}=1\big\}$,
which implies a self-financed investment without external borrowing.
The return of a portfolio $b\in\mathcal{B}^{m}$ with respect to $x_{n}$
is denoted by $\left\langle b,x_{n}\right\rangle $, where $\left\langle \cdot,\cdot\right\rangle $
denotes the scalar product of two vectors. Here, the assumption that
$X_{n}$ is bounded in $\mathcal{M}$ is reasonable, due to the common-sense
notion that assets' returns are unlikely to experience sudden and
significant increases between consecutive time periods in the real
market.\smallskip

Since the investor can causally make decisions based on observations,
the portfolio selection at any time $n$ forms a function of past
realizations $x_{1}^{n-1}$ and is denoted by $b_{n}:\mathbb{R}_{++}^{m\times(n-1)}\to\mathcal{B}^{m}$.
Following these portfolios, let us denote a \textit{strategy} associated
with them by the infinite sequence $\big(b_{n}\big)\coloneqq\left\{ b_{n}\right\} _{n=1}^{\infty}$;
and if the strategy has a fixed portfolio $b$ over time, it is called
a \textit{constant strategy} and denoted by $\left(b\right)$ without
a time index in order to distinguish it from a dynamic one. Since
short selling is prohibited, the investor's capital cannot be entirely
depleted, which would otherwise halt the investment; therefore, we
can define a performance measure for a strategy at any time period.
Let the initial fortune be $S_{0}\big(b_{0}\big)\eqqcolon S_{0}=1$
by convention, and assume the portfolios are made without commission
fees for arbitrary fractions.\smallskip 

With the above settings, the \textit{cumulative wealth} and its corresponding
exponential\textit{ }\textit{\emph{average}}\textit{ growth rate}
after $n$ periods of investment yielded by a strategy $\left(b_{n}\right)$
are respectively defined as:
\[
S_{n}\left(b_{n}\right)\coloneqq{\displaystyle {\displaystyle \prod_{i=1}^{n}}\left\langle b_{i},x_{i}\right\rangle }\text{ and }W_{n}\left(b_{n}\right)\coloneqq\dfrac{1}{n}\log S_{n}\left(b_{n}\right)={\displaystyle \dfrac{1}{n}\sum_{i=1}^{n}\log\left\langle b_{i},x_{i}\right\rangle .}
\]
Here, the notations $S_{n}\left(b_{n}\right)$ and $W_{n}\left(b_{n}\right)$
serve as respective shorthands for $S_{n}\big(\left\{ b_{i}\right\} _{i=1}^{n},x_{1}^{n}\big)$
and $W_{n}\big(\left\{ b_{i}\right\} _{i=1}^{n},x_{1}^{n}\big)$,
given the sequence of realizations $x_{1}^{n}$. Using the same shorthand,
the \textit{empirical expected value}, \textit{empirical variance},
and \textit{empirical Sharpe ratio} of return of the strategy $\left(b_{n}\right)$,
with the return $r\geq0$ of a reference portfolio, after $n\geq2$
periods, are defined respectively as:
\[
M_{n}\left(b_{n}\right)\coloneqq\frac{1}{n}\sum_{i=1}^{n}\left\langle b_{i},x_{i}\right\rangle ,\text{ }V_{n}\left(b_{n}\right)\coloneqq\frac{1}{n}\sum_{j=1}^{n}\Big(\big<b_{j},x_{j}\big>-\frac{1}{n}\sum_{i=1}^{n}\big<b_{i},x_{i}\big>\Big)^{2}\text{ and }Sh_{n}\left(b_{n}\right)\coloneqq\frac{M_{n}\left(b_{n}\right)-r}{\sqrt{V_{n}\left(b_{n}\right)}}.
\]
However, since the value of $r$ does not affect the analysis in the
following sections, we set $r=0$ by default henceforth to simplify
notation in subsequent discussions.

\subsection{Empirical Mean-Variance utility of a strategy and its consistency}

We consider a scenario where an investor constructs an M-V portfolio
with a predetermined risk aversion coefficient, implying a desired
personal risk profile, for long-term investment. The investor is assumed
to have an immutable attitude about the expected return-risk trade-off
for each distribution of asset returns. In other words, if one specific
distribution is given at different periods, the investor always calibrates
the same risk aversion coefficient. Therefore, this coefficient remains
fixed over time during the investment if the investor believes in
a hypothetical data distribution. However, since a hypothesis about
the true distribution is doubtful and results in an ambiguous risk
aversion, the investor could adjust the risk aversion coefficients
of the M-V model over time by learning new market information rather
than maintaining an unshakable belief.\smallskip

Let's imagine an unsophisticated investor collects a sample of historical
data, then derives a constant M-V strategy, i.e., an M-V portfolio,
with estimations for the parameters of the expected values and covariance
matrices of assets' returns of the M-V model, given a predetermined
risk aversion $\alpha$ depending on the distribution of the sample.
By doing so, this investor believes that the out-of-sample data, i.e.,
the future realizations of assets' returns, would have the same distribution
as that of the sample; unfortunately, this is often not the case as
future data are unknown beforehand but gradually revealed while being
difficult to predict correctly. Instead, with the learning ability,
another sophisticated investor should construct a dynamic strategy
$\big(b_{n}\big)$, rather than a constant strategy one, so that the
strategy $\big(b_{n}\big)$ could achieve the empirical performance
of the constant M-V strategy $\big(b^{\alpha}\big)$ associated with
a risk aversion $\alpha$, as if $\big(b^{\alpha}\big)$ is derived
using the M-V model with perfect knowledge of the distributions of
future assets' returns.\smallskip

In specific, given an infinite sequence $x_{1}^{\infty}$ as data
of all assets' returns realizations, let's define the corresponding
sequence of \textit{empirical distributions} $\big\{ P_{n}\big\}_{n=1}^{\infty}$,
with each $P_{n}$ denoting the empirical distribution of the assets'
returns variable $X$ at time $n$, as follows:
\[
P_{n}\left(D\right)\coloneqq\dfrac{1}{n}\sum_{i=1}^{n}\mathbb{I}_{x_{i}}\left(D\right),\,\forall D\subseteq\mathcal{M},\forall n,
\]
where $\mathbb{I}{}_{x}(\cdot)$ denotes the indicator function. Notably,
the empirical distribution does not require presumed statistical properties
by its definition. Besides, let's define further the following function
$\mathcal{L}\big(\alpha,b,\mu,\Sigma\big)$ as the M-V utility of
a portfolio $b$ in the simplex $\mathcal{B}^{m}$ given two parameters
of moments $\mu^{Q},\Sigma^{Q}$, corresponding to a distribution
$Q$, and a risk aversion coefficient $\alpha$:
\begin{equation}
\mathcal{L}\big(\alpha,b,\mu^{Q},\Sigma^{Q}\big)\coloneqq\mathbb{E}^{Q}\big(\big<b,X\big>\big)-\alpha\Var^{Q}\big(\big<b,X\big>\big)=\big<b,\mu^{Q}\big>-\alpha\big<b,\Sigma^{Q}b\big>,\,\alpha\geq0.\label{M-V model}
\end{equation}
From here on, the risk aversion coefficient is denoted by the non-negative
real value $\alpha$.\smallskip

With the above settings, for an infinite sequence $x_{1}^{\infty}$
and a corresponding risk aversion coefficient $\alpha\big(x_{1}^{\infty}\big)$
determined by the investor, the \textit{empirical M-V utility} $M_{n}\left(b_{n}\right)-\alpha\big(x_{1}^{\infty}\big)V_{n}\left(b_{n}\right)$
of a dynamic strategy $\big(b_{n}\big)$ is formally said to be\emph{
}\textit{\emph{(asymptotically)}}\textit{ consistent} with the utility
of the constant M-V strategies derived using the knowledge of the
distribution of $x_{1}^{\infty}$ if\footnote{The consistency (\ref{consistency}) is formalized with respect to
the standard online learning objective in the repeated game theoretical
frameworks such as \textit{sequential predictions} in \citet{CesaBianchi1997,CesaBianchi1999,CesaBianchi2006}
and \textit{online convex optimization} or \textit{online portfolio
optimization} as in \citet{ShalevShwartz2012}, \citet{Ordentlich1996},
\citet{Helmbold1998}, \citet{Blum1999}, \citet{Fostera1999},\textbf{
}\citet{Vovk1998}, \citet{Erven2020}, etc. However, thus far, this
framework has only considered the logarithmic utility, i.e., the growth
rate, as the objective for the portfolio selection problem. In contrast,
this paper proposes the M-V utility as the primary objective, and
even simultaneously incorporates the Sharpe ratio and growth rate.}:
\begin{align}
\lim_{n\to\infty}\min_{b\in\mathcal{B}^{m}}\big(M_{n}\big(b_{n}\big)-\alpha\big(x_{1}^{\infty}\big)V_{n}\big(b_{n}\big)-\mathcal{L}\big(\alpha\big(x_{1}^{\infty}\big),b,\mu^{P_{n}},\Sigma^{P_{n}}\big)\big)=0.\label{consistency}
\end{align}
Here, noting that the notation $\alpha\big(x_{1}^{\infty}\big)$ implies
that the true coefficient varies depending on each instance of a sequence
$x_{1}^{\infty}$, which is only determined if the whole sequence
is known to the investor. In addition, since the investor is concerned
with the empirical distribution of the data, all sequences that differ
merely by permutations of their values share a common coefficient.\smallskip
\begin{rem*}
The formalized consistency of the empirical M-V utility of a dynamic
strategy in (\ref{consistency}) does not necessarily imply that the
dynamic strategy will also achieve the other empirical performances
of the constant M-V strategy. As mentioned earlier, an investor adopts
the M-V model for its risk control capability, aiming to create an
M-V portfolio with a desired risk profile. Therefore, a dynamic strategy
is more appreciated if not only its empirical M-V utility is consistent
according to (\ref{consistency}), but also its empirical Sharpe ratio
and growth rate asymptotically approach those of the desired constant
M-V strategy as time progresses.
\end{rem*}

\section{Strategies and risk aversion calibration in deterministic market }

In this section, we propose algorithms to construct a strategy that
does not need forecasting but uses solely available data at each time,
ensuring consistency in empirical M-V utility, as in (\ref{consistency}),
and even achieves the desired empirical Sharpe ratio and growth rate
of the constant M-V strategy, which is derived knowing the whole infinite
sequence $x_{1}^{\infty}$. Furthermore, we also discuss the determination
of risk aversion under the ambiguity of the future distribution of
realizations of assets' returns, as well as the potential impact of
risk aversion calibration on the empirical performances of the proposed
strategies. For simplification, we intentionally treat the infinite
sequence of realizations as deterministic but unknown data throughout
the section.\smallskip

With the notation of the function $\mathcal{L}\big(\alpha,b,\mu^{Q},\Sigma^{Q}\big)$
in (\ref{M-V model}), which represents the M-V utility of a portfolio
$b$ with the parameters $\mu^{Q},\Sigma^{Q}$ corresponding to a
distribution $Q$ and a risk aversion $\alpha$, let us define the
corresponding set of M-V portfolios $B\left(\alpha,Q\right)$ as follows:
\begin{equation}
B\left(\alpha,Q\right)\coloneqq\big\{ b^{*}\in\mathcal{B}^{m}:\,\mathcal{L}\big(\alpha,b^{*},\mu^{Q},\Sigma^{Q}\big)\geq\mathcal{L}\big(\alpha,b,\mu^{Q},\Sigma^{Q}\big),\forall b\in\mathcal{B}^{m}\big\}.\label{Set B_n}
\end{equation}
Let $\mathcal{P}$ represent the space of all probability distributions
on the variable of assets' returns $X$. It is worth noting that since
the assets' returns are bounded, as $X\in\mathcal{M}$ as assumed
in the section on model settings, the first and second moments of
the variable $X$ corresponding to any distribution $Q\in\mathcal{P}$
are well defined, as is the M-V utility. The following Lemma \ref{Lemma 1}
establishes general properties of the set of optimizers $B\left(\alpha,Q\right)$
and the corresponding optimum $\max_{b\in\mathcal{B}^{m}}\mathcal{L}\big(\alpha,b,\mu^{Q},\Sigma^{Q}\big)$
in relation to the distribution $Q$ within a metrizable space $\mathcal{P}$,
which is useful for constructing the subsequent potential dynamic
strategy.\smallskip
\begin{lem}
Consider the M-V utility $\mathcal{L}\big(\alpha,b,\mu^{Q},\Sigma^{Q}\big)$
and the related set $B\left(\alpha,Q\right)$ of M-V portfolios corresponding
to a distribution of assets' returns $Q\in\mathcal{P}$. Given an
$\alpha$, we have:\label{Lemma 1}
\begin{enumerate}
\item The set $B\left(\alpha,Q\right)$ is non-empty, compact, and convex
for any $Q\in\mathcal{P}$.
\item For any sequence $\big\{ Q_{n}\big\}_{n=1}^{\infty}$ such that $Q_{n}\to Q_{\infty}$
weakly, the Hausdorff distance between the two sets $d_{H}\left(B\left(\alpha,Q_{n}\right),B\left(\alpha,Q_{\infty}\right)\right)\to0$
as $n\to\infty$, where: 
\[
d_{H}\left(B\big(\alpha,Q_{n}\big),B\big(\alpha,Q_{\infty}\big)\right)\coloneqq\max\Big\{\sup_{\bar{b}\in B(\alpha,Q_{n})}d\big(\bar{b},B\big(\alpha,Q_{\infty}\big)\big),\sup_{\hat{b}\in B(\alpha,Q_{\infty})}d\big(B\big(\alpha,Q_{n}\big),\hat{b}\big)\Big\},\vspace{-1ex}
\]
with $d\left(b,B\right)\coloneqq\inf_{\ddot{b}\in B}\left\Vert b-\ddot{b}\right\Vert $
for $B\subseteq\mathcal{B}^{m}$.
\item The function $\max_{b\in\mathcal{B}^{m}}\mathcal{L}\big(\alpha,b,\mu^{Q},\Sigma^{Q}\big)$
is continuous in $Q$ when the space $\mathcal{P}$ is equipped with
the weak topology. Consequently, the set $B\left(\alpha,Q\right)$
is closed in $\mathcal{B}^{m}\times\mathcal{P}$, and the function
of the optimal selector $b^{*}\left(\alpha,Q\right):\mathcal{P}\to\mathcal{B}^{m}$,
which satisfies $b^{*}\left(\alpha,Q\right)\in B\left(\alpha,Q\right)$,
is continuous at any $Q$ for which $B\left(\alpha,Q\right)$ is a
singleton set.
\end{enumerate}
\end{lem}
\begin{proof}
Since the quadratic objective function $\mathcal{L}\big(\alpha,b,\mu^{Q},\Sigma^{Q}\big)$
is concave over the convex and compact feasible set $\mathcal{B}^{m}$,
the respective set of the M-V portfolios $B\left(\alpha,Q\right)$
is a non-empty and convex subset of $\mathcal{B}^{m}$. Moreover,
for any sequence $\left\{ b_{n}\right\} _{n=1}^{\infty}$ such that
$b_{n}\to b_{\infty}$ and $b_{n}\in B\left(\alpha,Q\right)$ for
all $n$, we also have $b_{\infty}\in B\left(\alpha,Q\right)$ due
to the continuity of $\mathcal{L}\big(\alpha,b,\mu^{Q},\Sigma^{Q}\big)$
in $b\in\mathcal{B}^{m}$ given $\alpha$, $\mu^{Q}$, and $\Sigma^{Q}$,
as follows:
\begin{align*}
 & \mathcal{L}\big(\alpha,b_{n},\mu^{Q},\Sigma^{Q}\big)\geq\mathcal{L}\big(\alpha,b,\mu^{Q},\Sigma^{Q}\big),\forall b\in\mathcal{B}^{m},\forall n,\\
\Rightarrow & \lim_{n\to\infty}\mathcal{L}\big(\alpha,b_{n},\mu^{Q},\Sigma^{Q}\big)=\mathcal{L}\big(\alpha,b_{\infty},\mu^{Q},\Sigma^{Q}\big)\geq\mathcal{L}\big(\alpha,b,\mu^{Q},\Sigma^{Q}\big),\forall b\in\mathcal{B}^{m}.
\end{align*}
Hence, $B\left(\alpha,Q\right)$ is compact as it is closed and bounded
in $\mathbb{R}^{m}$, which proves (a).\smallskip

In order to prove assertion (b), we first show that $\sup_{\bar{b}\in B(\alpha,Q_{n})}d\big(\bar{b},B\big(\alpha,Q_{\infty}\big)\big)\to0$
by proving $d\big(\bar{b}_{n},B\big(\alpha,Q_{\infty}\big)\big)=\inf_{\ddot{b}\in B(\alpha,Q_{\infty})}\left\Vert \bar{b}_{n}-\ddot{b}\right\Vert \to0$,
as $Q_{n}\to Q_{\infty}$ weakly, for any sequence $\left\{ \bar{b}_{n}\right\} _{n=1}^{\infty}$
such that $\bar{b}_{n}\in B\left(\alpha,Q_{n}\right)$ for all $n$
(noting that the infimum attains a minimum due to the compactness
of the non-empty set $B\big(\alpha,Q_{\infty}\big)$, as stated in
assertion (a), which implies the existence of a minimizer). Assume
this is not true, so there is at least one sequence $\left\{ \bar{b}_{n}^{*}\right\} _{n=1}^{\infty}$
with $\bar{b}_{n}^{*}\in B\left(\alpha,Q_{n}\right)$ for all $n$,
such that there exists an $\epsilon^{*}>0$ and an $n^{N}>N$ for
any $N$, satisfying the following inequalities for $\bar{b}_{n^{N}}^{*}\in B\left(\alpha,Q_{n^{N}}\right)$
and any fixed choice of $\ddot{b}\in B\big(\alpha,Q_{\infty}\big)$:
\begin{equation}
\begin{cases}
\mathcal{L}\big(\alpha,\bar{b}_{n^{N}}^{*},\mu^{Q_{n^{N}}},\Sigma^{Q_{n^{N}}}\big)-\mathcal{L}\big(\alpha,\ddot{b},\mu^{Q_{n^{N}}},\Sigma^{Q_{n^{N}}}\big) & \geq0,\\
\mathcal{L}\big(\alpha,\ddot{b},\mu^{Q_{\infty}},\Sigma^{Q_{\infty}}\big)-\mathcal{L}\big(\alpha,\bar{b}_{n^{N}}^{*},\mu^{Q_{\infty}},\Sigma^{Q_{\infty}}\big) & \geq\epsilon^{*},
\end{cases}\label{equality proof}
\end{equation}
because $\inf_{\ddot{b}\in B(\alpha,Q_{\infty})}\left\Vert \bar{b}_{n^{N}}^{*}-\ddot{b}\right\Vert \geq\delta^{*}$
for some $\delta^{*}>0$, by the assumption of non-convergence and
the continuity of $\mathcal{L}\big(\alpha,b,\mu^{Q},\Sigma^{Q}\big)$
in $b$ (noting that it is possible that $\ddot{b}\in B\left(\alpha,Q_{n^{N}}\right)$).
Specifically, we define the compact set $D^{\delta^{*}}\big(\alpha,Q_{\infty}\big)\coloneqq\big\{ b\in\mathcal{B}^{m}:\,\inf_{\ddot{b}\in B(\alpha,Q_{\infty})}\left\Vert b-\ddot{b}\right\Vert \geq\delta^{*}\big\}$
for a given $\delta^{*}$, and assign the value $\epsilon^{*}$ as
$\epsilon_{D^{\delta^{*}}(\alpha,Q_{\infty})}^{*}$, given by:
\[
\max_{b\in\mathcal{B}^{m}}\mathcal{L}\big(\alpha,b,\mu^{Q_{\infty}},\Sigma^{Q_{\infty}}\big)-\max_{b\in D^{\delta^{*}}(\alpha,Q_{\infty})}\mathcal{L}\big(\alpha,b,\mu^{Q_{\infty}},\Sigma^{Q_{\infty}}\big)\eqqcolon\epsilon_{D^{\delta^{*}}(\alpha,Q_{\infty})}^{*}>0,
\]
where the latter term on the right-hand side decreases as $\delta^{*}$
increases due to its concavity.\smallskip

Hence, summing the two inequalities in (\ref{equality proof}) yields:
\begin{equation}
\begin{array}{ccc}
\mathcal{L}\big(\alpha,\bar{b}_{n^{N}}^{*},\mu^{Q_{n^{N}}},\Sigma^{Q_{n^{N}}}\big)-\mathcal{L}\big(\alpha,\bar{b}_{n^{N}}^{*},\mu^{Q_{\infty}},\Sigma^{Q_{\infty}}\big)\\
+\mathcal{L}\big(\alpha,\ddot{b},\mu^{Q_{\infty}},\Sigma^{Q_{\infty}}\big)-\mathcal{L}\big(\alpha,\ddot{b},\mu^{Q_{n^{N}}},\Sigma^{Q_{n^{N}}}\big) & \geq & \epsilon^{*},
\end{array}\label{equality proof 2}
\end{equation}
which contradicts the fact that there exists $N\big(\epsilon^{*},b\big)$
for all $b\in\mathcal{B}^{m}$ such that, by the triangle inequality,
the following sum of absolute terms must be strictly smaller than
$\epsilon^{*}$ for any $n>N\big(\epsilon^{*},b\big)$:
\[
\begin{array}{ccc}
\big\vert\mathcal{L}\big(\alpha,b,\mu^{Q_{n}},\Sigma^{Q_{n}}\big)-\mathcal{L}\left(\alpha,b,\mu^{Q_{\infty}},\Sigma^{Q_{\infty}}\right)\big\vert\\
+\big\vert\mathcal{L}\left(\alpha,\ddot{b},\mu^{Q_{\infty}},\Sigma^{Q_{\infty}}\right)-\mathcal{L}\left(\alpha,\ddot{b},\mu^{Q_{n}},\Sigma^{Q_{n}}\right)| & < & \epsilon^{*}.
\end{array}
\]
In other words, regarding (\ref{equality proof 2}), since $\epsilon^{*}$
is fixed and the latter term associated with the chosen $\ddot{b}$
decreases to zero over time, the former term is prevented from converging
to zero for some fixed portfolio choices. This follows from the convergence
$\mathcal{L}\left(\alpha,b,\mu^{Q_{n}},\Sigma^{Q_{n}}\right)\to\mathcal{L}\left(\alpha,b,\mu^{Q_{\infty}},\Sigma^{Q_{\infty}}\right)$
for any fixed $b$ and $\alpha$, which holds since $\mu^{Q_{n}}\to\mu^{Q_{\infty}}$
and $\Sigma^{Q_{n}}\to\Sigma^{Q_{\infty}}$ as $Q_{n}\to Q_{\infty}$
weakly. Consequently, any $\sup_{\bar{b}\in B(\alpha,Q_{n})}d\big(\bar{b},B\big(\alpha,Q_{\infty}\big)\big)$
can be attained by some optimal portfolios due to the compactness
of the non-empty sets $B\left(\alpha,Q_{n}\right)$, as stated in
assertion (a), and such a sequence $\left\{ \bar{b}_{n}^{*}\right\} _{n=1}^{\infty}$,
consisting of maximizers, must converge, as demonstrated.\smallskip

Using a similar contrapositive argument, we show that $\sup_{\hat{b}\in B(\alpha,Q_{\infty})}d\big(B\big(\alpha,Q_{n}\big),\hat{b}\big)\to0$
as $Q_{n}\to Q_{\infty}$ weakly by following the same procedure as
above, with the only difference being that the inequalities in (\ref{equality proof})
are swapped. Assuming there is a sequence $\left\{ \ddot{b}_{n}\right\} _{n=1}^{\infty}$
with $\ddot{b}_{n}\in B\left(\alpha,Q_{\infty}\right)$ for all $n$,
such that there exists a fixed $\delta^{*}>0$ and $n^{N}>N$ for
any $N$ satisfying $d\big(B\big(\alpha,Q_{n^{N}}\big),\ddot{b}_{n^{N}}\big)\geq\delta^{*}$,
then, for arbitrarily chosen $\bar{b}_{n^{N}}^{*}\in B\left(\alpha,Q_{n^{N}}\right)$,
we define $\epsilon_{n^{N}}^{*}$ as follows:
\[
\mathcal{L}\big(\alpha,\bar{b}_{n^{N}}^{*},\mu^{Q_{n^{N}}},\Sigma^{Q_{n^{N}}}\big)-\mathcal{L}\big(\alpha,\ddot{b}_{n^{N}},\mu^{Q_{n^{N}}},\Sigma^{Q_{n^{N}}}\big)\eqqcolon\epsilon_{n^{N}}^{*}>0,
\]
and subsequently obtain the following inequality:
\[
\begin{array}{ccc}
\big\vert\mathcal{L}\big(\alpha,\bar{b}_{n^{N}}^{*},\mu^{Q_{n^{N}}},\Sigma^{Q_{n^{N}}}\big)-\mathcal{L}\big(\alpha,\bar{b}_{n^{N}}^{*},\mu^{Q_{\infty}},\Sigma^{Q_{\infty}}\big)\big\vert\\
+\big\vert\mathcal{L}\big(\alpha,\ddot{b}_{n^{N}},\mu^{Q_{\infty}},\Sigma^{Q_{\infty}}\big)-\mathcal{L}\big(\alpha,\ddot{b}_{n^{N}},\mu^{Q_{n^{N}}},\Sigma^{Q_{n^{N}}}\big)\big\vert & \geq & \epsilon_{n^{N}}^{*}.
\end{array}
\]
Since $Q_{n^{N}}\to Q_{\infty}$ weakly as $N\to\infty$, both absolute
terms converge to zero, which results in $\epsilon_{n^{N}}^{*}\to0$.
Moreover, for a specific $\alpha$ and $Q_{n^{N}}$, because the objective
function $\mathcal{L}\left(\alpha,b,\mu^{Q_{n^{N}}},\Sigma^{Q_{n^{N}}}\right)$
is continuous in $b$ over the compact set $\mathcal{B}^{m}$, a near-optimality
implied by arbitrarily small $\epsilon_{n^{N}}^{*}$ ensures the proximity
of $\ddot{b}_{n}$ to the set of maximizers $B\left(\alpha,Q_{n^{N}}\right)$.
Therefore, this contradicts the assumed existence of a fixed $\delta^{*}>0$,
so $\sup_{\hat{b}\in B(\alpha,Q_{\infty})}d\big(B\big(\alpha,Q_{n}\big),\hat{b}\big)\to0$,
as required for assertion (b).\smallskip

To demonstrate assertion (c), since the variable $X$ is defined on
a compact subset $\mathcal{M}$ of the metric space $\mathbb{R}^{m}$,
the space $\mathcal{P}$, when equipped with the weak topology, is
also compact and metrizable (recalling that the weak topology is the
weakest topology on $\mathcal{P}$ such that $\mathbb{E}^{Q}\big(f\big(X\big)\big)$
is continuous in $Q$ for any bounded continuous real function $f(\cdot)$).
With this weakest topology, $\mu^{Q}$ and $\Sigma^{Q}$ are continuous
in $Q$, so the function $\mathcal{L}\left(\alpha,b,\mu^{Q},\Sigma^{Q}\right)$
is also continuous in $Q$. Therefore, for any $Q_{n}\to Q_{\infty}$
in $\mathcal{P}$, $d_{H}\left(B\left(\alpha,Q_{n}\right),B\left(\alpha,Q_{\infty}\right)\right)\to0$
by assertion (b), so $\max_{b\in\mathcal{B}^{m}}\mathcal{L}\big(\alpha,b,\mu^{Q_{n}},\Sigma^{Q_{n}}\big)\to\max_{b\in\mathcal{B}^{m}}\mathcal{L}\big(\alpha,b,\mu^{Q_{\infty}},\Sigma^{Q_{\infty}}\big)$,
which confirms that the function $\max_{b\in\mathcal{B}^{m}}\mathcal{L}\big(\alpha,b,\mu^{Q},\Sigma^{Q}\big)$
is continuous in $Q$. Also, by assertion (b) in this context, for
any sequence $\left\{ b_{n}\right\} _{n=1}^{\infty}$ such that $b_{n}\to b_{\infty}$
and $b_{n}\in B\big(\alpha,Q_{n}\big)$ for all $n$, we must have
$b_{\infty}\in B\big(\alpha,Q_{n}\big)$, which asserts the closedness
of the set $B\left(\alpha,Q\right)$ in $\mathcal{B}^{m}\times\mathcal{P}$.
Additionally, if the set $B\left(\alpha,Q^{*}\right)$ is a singleton
at $Q^{*}$, then the limiting point $b_{\infty}$ of any sequence
$\left\{ b_{n}\right\} _{n=1}^{\infty}$ is unique, which establishes
the continuity of the optimal selection function $b^{*}\left(\alpha,Q\right)$
at $Q^{*}$. Thus, all the statements are proved, completing the proof.\smallskip
\end{proof}
\textbf{Summability condition}. Prior to proposing the strategy, we
impose a minimal summability condition on the sequences of realizations.
The philosophical justification for this condition is that the M-V
model implies a utility function with respect to the true distribution
of future data, which is perceived through the single risk aversion
of the immutable-attitude investor. Thus, although the empirical distributions
of future data change over time, they should not change infinitely
often over an unlimited horizon. Specifically, let us assume that
for a deterministic infinite sequence $x_{1}^{\infty}$, the corresponding
sequence of empirical distributions $\big\{ P_{n}\big\}_{n=1}^{\infty}$
converges weakly to a limiting distribution $P_{\infty}$ as the number
of observations $n\to\infty$. In this context, if the sequence $x_{1}^{\infty}$
is entirely known to the investor, the corresponding limiting distribution
$P_{\infty}$ is also known, and thus the associated risk aversion
$\alpha\big(P_{\infty}\big)\coloneqq\alpha\big(x_{1}^{\infty}\big)$
is uniquely determined regardless of permutations of $x_{1}^{\infty}$.
Consequently, the consistency according to (\ref{consistency}) becomes
the convergence of the empirical M-V utility of a strategy $\big(b_{n}\big)$
to the limiting M-V utility $\max_{b\in\mathcal{B}^{m}}\mathcal{L}\big(\alpha\big(P_{\infty}\big),b,\mu^{P_{\infty}},\Sigma^{P_{\infty}}\big)$,
since Lemma \ref{Lemma 1} asserts that $\max_{b\in\mathcal{B}^{m}}\mathcal{L}\big(\alpha,b,\mu^{P_{n}},\Sigma^{P_{n}}\big)$$\to\max_{b\in\mathcal{B}^{m}}\mathcal{L}\big(\alpha,b,\mu^{P_{\infty}},\Sigma^{P_{\infty}}\big)$
for any coefficient $\alpha$.\smallskip

\subsection{Online estimation algorithm for market without redundant assets}

The following algorithm for constructing a dynamic strategy simply
updates the estimations over time for the parameters of the M-V model\footnote{Updating estimators here does not necessarily imply a Bayesian model,
which will be discussed in section 4, since the investor is not supposed
to estimate the conditional expectation $\mathbb{E}\big(\left\langle b,X_{n}\right\rangle |X_{1}^{n-1}\big)$
and variance $\Var\big(\left\langle b,X_{n}\right\rangle |X_{1}^{n-1}\big)\big)$
at each time $n$. Additionally, although there exist practical examples
of risk aversion coefficients for investors referenced in \citet{Grossman1981,Bodie2005,Ang2014},
they are not the same for all individuals and generally cannot be
imposed in all markets as their distributions of assets' returns change
and remain unknown.}, based on observations of past realizations during the investment,
as the functions $\mu_{n}\coloneqq\mu\big(x_{1}^{n-1}\big)$ and $\varSigma_{n}\coloneqq\varSigma\big(x_{1}^{n-1}\big)$.
Given a predetermined risk aversion coefficient $\alpha$, the investor
derives an M-V portfolio corresponding to the M-V model with updated
parameters at each time $n$. In the remainder of this paper, all
dynamic strategies of this kind are collectively referred to as the
(dynamic) M-V strategy.\smallskip

\textbf{Strategy proposal with constant risk aversion}. For a predetermined
risk aversion $\alpha$, the investor causally derives a portfolio
of the M-V strategy $\big(b_{n}^{\alpha}\big)$ at each time $n$
as follows:
\begin{equation}
b_{n}^{\alpha}\coloneqq\left(1/m,...,1/m\right),\forall n\leq h\text{ and }b_{n}^{\alpha}\in B\left(\alpha,P_{n-1}\right),\forall n>h,\label{Proposal strategy mean variance}
\end{equation}
where $B\left(\alpha,P_{n}\right)$ is defined according to (\ref{Set B_n})
and the initial $b_{n}^{\alpha}$ for $n\leq h$ can be chosen arbitrarily
rather than equal allocation as in (\ref{Proposal strategy mean variance}).
The number $h$ is a manual input which should be significantly larger
than the number of assets $m$ but not too large, for instance $h=2m$.
Noteworthily, in this construction, the investor is simply using the
estimators $\mu_{n}=\mu^{P_{n-1}}$ and $\Sigma_{n}=\Sigma^{P_{n-1}}$
for the parameters of the M-V model, depending on past realizations
$x_{1}^{n-1}$. In a market with no redundant assets according to
Definition \ref{Redundant}, the consistency of the proposed M-V strategy
$\big(b_{n}^{\alpha}\big)$ is asserted by the following Theorem \ref{Theorem 1}
for any predetermined risk aversion.
\begin{defn}
A market of $m$ risky assets with returns variable $X\sim Q$ is
said not to include redundant assets, concerning the distribution
$Q\in\mathcal{P}$, if the covariance matrix $\Sigma^{Q}$ is (strictly)
positive definite. Accordingly, when $P_{n}\to P_{\infty}$, the empirical
market is said not to include redundant assets if $\Sigma^{P_{\infty}}$
is a positive definite matrix.\label{Redundant}
\end{defn}
\begin{rem*}
As commonly referred to in the finance literature, especially in portfolio
selection, redundant assets are defined as ones that are linearly
dependent on the remaining assets, thereby not benefiting investment
portfolio diversification. For an infinite sequence $x_{1}^{\infty}$,
when the number of observations $n$ increases to be somewhat larger
than the number of assets $m$, the empirical distributions $P_{n}$
can have support sets that are full-dimensional, even if the involved
assets exhibit strong correlation; thus, the corresponding empirical
covariance matrices $\Sigma^{P_{n}}$ are also positive definite due
to their full rank. Consequently, when the empirical market does not
include redundant assets, the investor will observe positive definite
covariance matrices as time evolves without needing to know the limiting
distribution.\smallskip
\end{rem*}
\begin{lem}
Given an infinite sequence of realizations $x_{1}^{\infty}$ with
the corresponding sequence of empirical distributions $\big\{ P_{n}\big\}_{n=1}^{\infty}$,
consider two generic strategies $\big(\hat{b}_{n}\big)$ and $\big(\bar{b}_{n}\big)$,
i.e., they can be constant, such that $\big(\hat{b}_{n}-\bar{b}_{n}\big)\to0$.
Then, as $n\to\infty$, all the following sequences of differences
simultaneously converge to zero:
\[
\big\{ M_{n}\big(\hat{b}_{n}\big)-M_{n}\big(\bar{b}_{n}\big)\big\}_{n=1}^{\infty},\text{ }\big\{ V_{n}\big(\hat{b}_{n}\big)-V_{n}\big(\bar{b}_{n}\big)\big\}_{n=1}^{\infty},\text{ and }\big\{ W_{n}\big(\hat{b}_{n}\big)-W_{n}\big(\bar{b}_{n}\big)\big\}_{n=1}^{\infty},
\]
where the last term requires\textbf{ $\big<\hat{b}_{n},x_{n}\big>$
}and\textbf{ $\big<\bar{b}_{n},x_{n}\big>$ }to be bounded away from
zero for all large $n$.\label{lemma 2}
\end{lem}
\begin{proof}
The convergence $\big(\hat{b}_{n}-\bar{b}_{n}\big)\to0$ implies $\big(\big<\hat{b}_{n},x_{n}\big>-\big<\bar{b}_{n},x_{n}\big>\big)\to0$
due to:
\[
|\big<\hat{b}_{n}-\bar{b}_{n},x_{n}\big>|\leq\|\hat{b}_{n}-\bar{b}_{n}\|\|x_{n}\|\leq M\|\hat{b}_{n}-\bar{b}_{n}\|,\,\forall n,
\]
by the Cauchy\textendash Schwarz inequality and the boundedness of
the sequence $x_{1}^{\infty}$ (as $X_{n}\in\mathcal{M}$). Then:
\begin{align}
\lim_{n\to\infty}\big(\big<\hat{b}_{n},x_{n}\big>-\big<\bar{b}_{n},x_{n}\big>\big)=0 & \Rightarrow\lim_{n\to\infty}\frac{1}{n}\sum_{i=1}^{n}\big(\big<\hat{b}_{i},x_{i}\big>-\big<\bar{b}_{i},x_{i}\big>\big)=0\nonumber \\
 & \Rightarrow\lim_{n\to\infty}\big(M_{n}\big(\hat{b}_{n}\big)-M_{n}\big(\bar{b}_{n}\big)\big)=0,\label{M_n - M_n}
\end{align}
due to the Cesaro mean theorem. Since $\big(\log\big<\hat{b}_{n},x_{n}\big>-\log\big<\bar{b}_{n},x_{n}\big>\big)\to0$
due to $\big<\hat{b}_{n},x_{n}\big>\big/\big<\bar{b},x_{n}\big>\to1$,
by using a similar argument as above, we obtain the next needed result
as:
\[
\lim_{n\to\infty}\big(W_{n}\big(\hat{b}_{n}\big)-W_{n}\big(\bar{b}_{n}\big)\big)=0,
\]
given that both\textbf{ $\big<\hat{b}_{n},x_{n}\big>$ }and\textbf{
$\big<\bar{b}_{n},x_{n}\big>$} are assumed to be bounded away from
zero for all large $n$.\smallskip

On the other hand, by using the obtained results in (\ref{M_n - M_n}),
we have the following limit:
\begin{equation}
\lim_{n\to\infty}\Big(\Big(\big<\hat{b}_{n},x_{n}\big>-\frac{1}{n}\sum_{i=1}^{n}\big<\hat{b}_{i},x_{i}\big>\Big)^{2}-\Big(\big<\bar{b}_{n},x_{n}\big>-\frac{1}{n}\sum_{i=1}^{n}\big<\bar{b}_{i},x_{i}\big>\Big)^{2}\Big)=\lim_{n\to\infty}A_{n}B_{n}=0,\label{V_n-V_n lim}
\end{equation}
where 
\[
B_{n}\coloneqq\Big(\big<\hat{b}_{n},x_{n}\big>-\frac{1}{n}\sum_{i=1}^{n}\big<\hat{b}_{i},x_{i}\big>+\big<\bar{b}_{n},x_{n}\big>-\frac{1}{n}\sum_{i=1}^{n}\big<\bar{b}_{i},x_{i}\big>\Big)
\]
is bounded as $\big<\hat{b}_{n},x_{n}\big>$ and $\big<\bar{b},x_{n}\big>$
are bounded over the compact set $\mathcal{B}^{m}$ for any $x_{n}$,
so both $n^{-1}\sum_{i=1}^{n}\big<\hat{b}_{i},x_{i}\big>$ and $n^{-1}\sum_{i=1}^{n}\big<\bar{b}_{i},x_{i}\big>$
are bounded for all $n$ as well. Meanwhile:
\[
\lim_{n\to\infty}A_{n}\coloneqq\lim_{n\to\infty}\Big(\big<\hat{b}_{n},x_{n}\big>-\big<\bar{b}_{n},x_{n}\big>+\frac{1}{n}\sum_{i=1}^{n}\big<\bar{b}_{i},x_{i}\big>-\frac{1}{n}\sum_{i=1}^{n}\big<\hat{b}_{i},x_{i}\big>\Big)=0.
\]
Then, by invoking the Cesaro mean theorem again for (\ref{V_n-V_n lim}),
we obtain the last needed result:
\begin{align*}
 & \lim_{n\to\infty}\frac{1}{n}\sum_{j=1}^{n}\Big(\Big(\big<\hat{b}_{j},x_{j}\big>-\frac{1}{n}\sum_{i=1}^{n}\big<\hat{b}_{i},x_{i}\big>\Big)^{2}-\Big(\big<\bar{b}_{j},x_{j}\big>-\frac{1}{n}\sum_{i=1}^{n}\big<\bar{b}_{i},x_{i}\big>\Big)^{2}\Big)=0\\
\Rightarrow & \lim_{n\to\infty}\big(V_{n}\big(\hat{b}_{n}\big)-V_{n}\big(\bar{b}_{n}\big)\big)=0,
\end{align*}
and this completes the proof.
\end{proof}
\begin{thm}
For an infinite sequence of realizations $x_{1}^{\infty}$, assume
the corresponding empirical distribution converges weakly to a limiting
distribution as $P_{n}\to P_{\infty}$ and the empirical market does
not include redundant assets with respect to $P_{\infty}$. Consider
the M-V strategy $\big(b_{n}^{\alpha}\big)$ constructed according
to (\ref{Proposal strategy mean variance}) with any constant risk
aversion $\alpha$, then:\label{Theorem 1}
\begin{enumerate}
\item $\lim_{n\to\infty}\big(M_{n}\left(b_{n}^{\alpha}\right)-\alpha V_{n}\left(b_{n}^{\alpha}\right)\big)=\mathcal{L}\big(\alpha,b^{\alpha},\mu^{P_{\infty}},\Sigma^{P_{\infty}}\big),\vspace{0.2ex}$
\item $\lim_{n\to\infty}Sh_{n}\left(b_{n}^{\alpha}\right)={\displaystyle \frac{\mathbb{E}^{P_{\infty}}\big(\big<b^{\alpha},X\big>\big)}{\sqrt{\smash[b]{\Var^{P_{\infty}}\big(\big<b^{\alpha},X\big>\big)}}}}\text{ and }\lim_{n\to\infty}W_{n}\left(b_{n}^{\alpha}\right)=\mathbb{E}^{P_{\infty}}\big(\log\big<b^{\alpha},X\big>\big),$
\end{enumerate}
with $b^{\alpha}\in B\big(\alpha,P_{\infty}\big)$, which is a singleton
set.
\end{thm}
\begin{proof}
Since the empirical market does not include redundant assets, the
covariance matrix $\Sigma^{P_{\infty}}$ is positive definite, so
the function $\mathcal{L}\big(\alpha,b^{\alpha},\mu^{P_{\infty}},\Sigma^{P_{\infty}}\big)$
is strictly concave over the compact and convex feasible set $\mathcal{B}^{m}$.
Thus, the corresponding set $B\big(\alpha,P_{\infty}\big)$ is a singleton
set and by Lemma \ref{Lemma 1}, we have $b_{n}^{\alpha}\to b^{\alpha}\in B\big(\alpha,P_{\infty}\big)$.
Then, by Lemma \ref{lemma 2} we have the following limits:
\[
\begin{cases}
{\displaystyle \lim_{n\to\infty}M_{n}\left(b_{n}^{\alpha}\right)}={\displaystyle \lim_{n\to\infty}\mathbb{E}^{P_{n}}\big(\big<b^{\alpha},X\big>\big)}=\mathbb{E}^{P_{\infty}}\big(\big<b^{\alpha},X\big>\big),\\
{\displaystyle \lim_{n\to\infty}V_{n}\left(b_{n}^{\alpha}\right)}={\displaystyle \lim_{n\to\infty}\Var^{P_{n}}\big(\big<b^{\alpha},X\big>\big)}=\Var^{P_{\infty}}\big(\big<b^{\alpha},X\big>\big),
\end{cases}
\]
due to $P_{n}\to P_{\infty}$ weakly, and the limiting point $b^{\alpha}$
being treated as a constant M-V strategy $\big(b^{\alpha}\big)$,
so the assertion (a) follows immediately. For the assertion (b), the
convergence of the empirical Sharpe ratio of return is a trivial consequence
of the convergence of the sequences $\left\{ M_{n}\left(b_{n}^{\alpha}\right)\right\} _{n=1}^{\infty}$
and $\left\{ V_{n}\left(b_{n}^{\alpha}\right)\right\} _{n=1}^{\infty}$,
while the convergence of the growth rate sequence $\left\{ W_{n}\left(b_{n}^{\alpha}\right)\right\} _{n=1}^{\infty}$
is a direct result of Lemma \ref{lemma 2} and a similar argument
as above.
\end{proof}
Theorem \ref{Theorem 1} affirms that a constant M-V strategy, derived
before starting the investment using a single collected sample for
estimating model parameters, will experience the desired performances
only if the empirical distributions of future realizations converge
to the distribution of the training dataset. For instance, regarding
the traditional assumption in the literature of M-V theory, if the
collected sample of assets' returns is realizations of an i.i.d distribution,
significantly large enough, then the empirical distributions of future
data will almost surely converge to the same distribution of the sample,
as per the strong law of large numbers or the Glivenko\textendash Cantelli
theorem. It's important to note that the i.i.d. assumption is much
less general than the condition of weak convergence, as a converging
sequence of empirical distributions does not necessarily imply i.i.d.,
and the distribution of the future data may differ from that of the
collected sample. Consequently, misestimation leads to underperformance
for any constant M-V strategy compared to those derived with knowledge
of the limiting distribution.\smallskip

Moreover, if the risk aversion corresponding to the limiting distribution
can be calibrated beforehand as $\alpha\big(P_{\infty}\big)$, the
M-V strategy $\big(b_{n}^{\alpha(P_{\infty})}\big)$ yields the empirical
M-V utility that is consistent in the sense of (\ref{consistency}),
while also asymptotically achieving the desired Sharpe ratio and growth
rate of the constant M-V strategy as the unique portfolio in $B\big(\alpha\big(P_{\infty}\big),P_{\infty}\big)$.
Unfortunately, neither the distribution $P_{\infty}$ nor the risk
aversion $\alpha\big(P_{\infty}\big)$ is known beforehand. If the
investor can define a risk aversion updating function $\alpha\big(Q\big)$
that is continuous in $Q\in\mathcal{P}$, then the M-V strategy $\big(b_{n}^{\alpha_{n}}\big)$,
derived according to (\ref{Proposal strategy mean variance}) with
updating $\alpha_{n}=\alpha\big(P_{n-1}\big)$ instead of a constant
coefficient, will solve the problem; nonetheless, such a continuous
risk aversion function seems to be nontrivial and challenging to construct.
In the subsequent sections, we will discuss the impact of risk aversion
on the performances of M-V strategies and propose a method of coefficient
updating to ensure either the optimal Sharpe ratio or the optimal
growth rate under an unknown limiting distribution.

\subsection{Relation of growth and mean-variance tradeoff}

In this section, we discuss the impact of the investor's risk aversion
on the growth rate of the M-V strategies. It is easy to show that
$\mathbb{E}^{P_{n}}\left(\log\big<b,X\big>\right)\leq\log\mathbb{E}^{P_{n}}\left(\big<b,X\big>\right)<\mathbb{E}^{P_{n}}\big(\big<b,X\big>\big)$
for all empirical distributions $P_{n}$, by using Jensen's inequality
or the AM\textendash GM inequality. Thus, a constant strategy with
a higher empirical expected return does not guarantee a higher empirical
growth rate. However, since we are considering the M-V strategies,
the quadratic utility function should be regarded rather than simply
the expected return utility, which involves a tradeoff with the variance
of portfolio returns. Proposition \ref{Normal mean variance tradeoff}
provides a theoretical explanation for the relation between the Sharpe
ratio and the expected logarithmic utility of a portfolio along the
efficient frontier with a normal distribution of assets' returns,
which is typically considered in the finance literature and especially
in M-V theory. The remark below stresses that this potential relationship
has been an interesting open question for a long time.\smallskip
\begin{rem*}
The relation between the expected logarithmic return and M-V utility
has attracted research attention for decades, as they are two main
models in portfolio selection, which is empirically experienced in
the market with a normal distribution for assets' returns. For instance,
experiments with various distributions in \citet{Grauer1981} show
that the strategies of the two models yield similar performances if
the assets' returns follow a normal distribution; such results are
also found empirically in research by \citet{Kroll1984,Levy1979}.
In a more specific case of a log-normal distribution, \citet{Merton1973}
shows that the portfolio with the optimal expected logarithmic return
is instantaneously mean-variance efficient under a continuous-time
model. Yet, the Sharpe ratio and risk aversion seem not to have been
considered for explanation, which the following proposition addresses
by showing that the portfolio maximizing expected logarithmic return
must also be M-V efficient for a given risk aversion, i.e., it lies
on the efficient frontier.
\end{rem*}
\begin{prop}
For two portfolios $b$, $\bar{b}$ and a normally distributed variable
$X\sim\mathcal{N}\left(\mu,\Sigma\right)$, if their expected returns
and Sharpe ratios satisfy, respectively:
\[
\mathbb{E}\big(\big<b,X\big>\big)\geq\mathbb{E}\big(\big<\bar{b},X\big>\big)\text{ and }\frac{\mathbb{E}\big(\big<b,X\big>\big)}{\sqrt{\smash[b]{\Var\left(\left\langle b,X\right\rangle \right)}}}\geq\frac{\mathbb{E}\big(\big<\bar{b},X\big>\big)}{\sqrt{\smash[b]{\Var\big(\big<\bar{b},X\big>\big)}}},
\]
then $\mathbb{E}\big(\log\big<b,X\big>\big)\geq\mathbb{E}\big(\log\big<\bar{b},X\big>\big)$,
with equality holding if the above conditions are equalities.\label{Normal mean variance tradeoff}
\end{prop}
\begin{proof}
Since $X\sim\mathcal{N}\left(\mu,\Sigma\right)$, we have $\big<b,X\big>\sim\mathcal{N}\left(\mu_{b},\sigma_{b}^{2}\right)$
for any portfolio $b\in\mathcal{B}^{m}$, where $\mu_{b}\coloneqq\left\langle b,\mu\right\rangle $
and $\sigma_{b}^{2}\coloneqq\big<b,\Sigma b\big>$ are strictly positive.
Then, by using the Taylor expansion around $\mu_{b}$ for $\mathbb{E}\left(\log\big<b,X\big>\right)$,
we have the following expression in terms of Sharpe ratio and $\mu_{b}$:
\begin{align*}
\mathbb{E}\left(\log\big<b,X\big>\right) & =\log\mathbb{E}\big(\big<b,X\big>\big)+\sum_{k=1}^{\infty}\frac{\big(-1\big)^{k+1}\mathbb{E}\big(\big<b,X\big>-\mathbb{E}\big(\big<b,X\big>\big)\big)^{k}}{k\big(\mathbb{E}\big(\big<b,X\big>\big)\big)^{k}}\\
 & =\log\mu_{b}+\sum_{k=1}^{\infty}\frac{\big(-1\big)^{k+1}\mathbb{E}\big(\big<b,X\big>-\mu_{b}\big)^{k}}{k\mu_{b}^{k}}\\
 & =\log\mu_{b}-\sum_{i=1}^{\infty}\frac{1^{2i+1}\left(2i-1\right)!!}{2i}\Big(\frac{\sigma_{b}}{\mu_{b}}\Big)^{2i}
\end{align*}
Here, due to the symmetry of the normally distributed $\left\langle b,X\right\rangle $,
all the odd-order central moments vanish, while the even-order central
moments satisfy $\mathbb{E}\big(\big<b,X\big>-\mu_{b}\big)^{2i}=\left(2i-1\right)!!\sigma_{b}^{2i}$
for all $i$. Hence, if $\mu_{b}\geq\mu_{\bar{b}}$ and $\mu_{b}/\sigma_{b}\geq\mu_{\bar{b}}/\sigma_{\bar{b}}$,
then $\mathbb{E}\left(\log\big<b,X\big>\right)\geq\mathbb{E}\big(\log\big<\bar{b},X\big>\big)$
as needed to show.
\end{proof}
According to Proposition \ref{Normal mean variance tradeoff}, a constant
M-V strategy will yield a higher growth rate than any other constant
strategy that has both a lower empirical Sharpe ratio and expected
return, provided the number of observed realizations of assets' returns
is sufficiently large and their empirical distribution has a symmetric
shape as that of a normal one. This occurs because all M-V portfolios,
derived corresponding to different risk aversion coefficients, lie
on the efficient frontier; thus, each always has a higher Sharpe ratio
compared to inferior portfolios located below the frontier curve,
which have the same values of either variance or expected return.
When the risk aversion coefficient $\alpha$ is lowered as the investor
favors the expected return over the risk of a portfolio, the expected
value tends to shift upward along the efficient frontier. Overall,
this results in a tradeoff between the expected return of the portfolio
and its variance, often leading to an increase in the Sharpe ratio
of the M-V portfolios along the upward frontier up to a certain point.
Therefore, as long as moving the M-V portfolio upward along the frontier
with increasing expected return and variance, but at a slower rate
for the latter, the growth rate and also the cumulative wealth of
the respective constant M-V strategy generally increase. Notably,
along the frontier curve, if the M-V portfolio crosses the one maximizing
the Sharpe ratio, the rate of increase of variance tends to be faster
than that of the expected return, as the slope of the tangent line
at such a point is lowered.\smallskip 

This relation between the expected logarithmic return and the mean-variance
tradeoff asserts a particular context where an investor favoring the
expected return, up to a certain extent, often outperforms other investors
who are more risk-averse in terms of cumulative wealth of constant
M-V strategy. Under this context, for any presumed risk aversion coefficient
$\alpha$, the corresponding M-V strategy $\big(b_{n}^{\alpha}\big)$,
formed according to (\ref{Proposal strategy mean variance}), guarantees
the same empirical growth rate and Sharpe ratio as that of the constant
M-V strategy with knowledge of the limiting normal distribution by
Theorem \ref{Theorem 1}. Nonetheless, since the parameters of the
limiting distribution are unknown, even if the bell curve of the empirical
distributions of data can be observed as time progresses, the investor
is ambiguous about determining an appropriate risk aversion in advance.
Also, an alternative updating risk aversion function that is continuous
in distribution is difficult to define. Instead, the following section
proposes a way of constructing a sequence of updating risk aversion
coefficients converging to the desired one corresponding to the limiting
distribution.

\subsection{Adaptive risk aversion for optimal Sharpe ratio and growth rate}

Although constructing an updating risk aversion function $\alpha\left(Q\right)$
that is continuous in $Q\in\mathcal{P}$ is challenging, the investor
could reduce it to the task of constructing a sequence of adaptive
risk aversions $\left\{ \alpha_{n}\right\} _{n=1}^{\infty}$ based
on the observed empirical distributions $P_{n-1}$ for the times $n$,
such that $\alpha_{n}\to\alpha\big(P_{\infty}\big)$ as $P_{n}\to P_{\infty}$.
Under the ambiguity of the true risk aversion $\alpha\big(P_{\infty}\big)$,
the investor should instead consider a set of several coefficients
below a conjectured threshold rather than a single one at any time,
so a dynamic M-V strategy associated with one among them can achieve
the correct one corresponding to the limiting distribution. We propose
below an algorithm with an updating coefficient $\alpha_{n}$, based
on past realizations, to guarantee attaining the optimal asymptotic
Sharpe ratio and growth rate among the constant M-V strategies with
respect to a set of risk aversions.\smallskip 

\textbf{Strategy proposal with updating risk aversion}. Given a set
of finitely many risk aversion coefficients $\mathcal{A}\subset\mathbb{R}_{+}$,
the respective M-V strategy $\big(b_{n}^{\alpha}\big)$ is constructed
according to (\ref{Proposal strategy mean variance}) for any fixed
coefficient $\alpha\in\mathcal{A}$ as the subroutine algorithm. Then,
the investor adaptively constructs the M-V strategy $\big(b_{n}^{\alpha_{n}}\big)$
according to the updating coefficients $\alpha_{n}$ at time $n$
as follows:
\begin{equation}
b_{n}^{\alpha_{n}}\coloneqq\left(1/m,...,1/m\right),\forall n\leq h\text{ and }b_{n}^{\alpha_{n}}\coloneqq\operatorname*{argmin}_{b\in\bar{B}(P_{n-1})}\|b-b_{n-2}^{\alpha_{n-2}}\|,\forall n>h,\label{proposal strategy optimal sharpe}
\end{equation}
where $\bar{B}\left(P_{n}\right)\coloneqq\big\{ b_{n}^{\bar{\alpha}}:\,\bar{\alpha}\in\mathcal{A}\text{ and }\mathcal{U}\big(b_{n}^{\bar{\alpha}},P_{n}\big)\geq\mathcal{U}\big(b_{n}^{\alpha},P_{n}\big),\forall\alpha\in\mathcal{A}\big\}$
for all $n\geq h$ with the function $\mathcal{U}\big(b,Q\big)$ chosen
either as Sharpe ratio or expected logarithmic return as below:
\begin{equation}
\mathcal{U}\big(b,Q\big)\coloneqq{\displaystyle \frac{\mathbb{E}^{Q}\big(\big<b,X\big>\big)}{\sqrt{\smash[b]{\Var^{Q}\big(\big<b,X\big>\big)}}}},\text{ or }\mathcal{U}\big(b,Q\big)\coloneqq\mathbb{E}^{Q}\big(\log\big<b,X\big>\big),\forall n.\smallskip\label{Function U}
\end{equation}
By this construction, the investor updates the risk aversion at each
period by selecting the optimal one from the previous period, regarding
a determined objective, without knowing the limiting distribution.
The guarantee of performance for the M-V strategy $\big(b_{n}^{\alpha_{n}}\big)$
is affirmed in Corollary \ref{Corollary 1}.
\begin{cor}
For an infinite sequence of realizations $x_{1}^{\infty}$ with similar
assumptions as in Theorem \ref{Theorem 1}, consider a set of risk
aversions $\mathcal{A}$ and the M-V strategy $\big(b_{n}^{\alpha_{n}}\big)$
constructed according to (\ref{proposal strategy optimal sharpe}).
Then, either one of the following limits holds:\label{Corollary 1}
\[
\lim_{n\to\infty}Sh_{n}\big(b_{n}^{\alpha_{n}}\big)=\frac{\mathbb{E}^{P_{\infty}}\big(\big<b^{\alpha^{*}},X\big>\big)}{\sqrt{\smash[b]{\Var^{P_{\infty}}\big(\big<b^{\alpha^{*}},X\big>\big)}}}{\displaystyle \text{ or }}\lim_{n\to\infty}W{}_{n}\big(b_{n}^{\alpha_{n}}\big)=\mathbb{E}^{P_{\infty}}\big(\log\big<b^{\alpha^{*}},X\big>\big),
\]
corresponding to the choice of the function $\mathcal{U}\big(b,Q\big)$
as in (\ref{Function U}), where:
\[
b^{\alpha^{*}}\in B\left(\alpha^{*},P_{\infty}\right),\alpha^{*}\in\mathcal{A}\text{ such that }\mathcal{U}\big(b^{\alpha^{*}},P_{\infty}\big)\geq\mathcal{U}\big(b^{\alpha},P_{\infty}\big),\forall b^{\alpha}\in B\left(\alpha,P_{\infty}\right),\forall\alpha\in\mathcal{A}.
\]
In other words, under the limiting distribution $P_{\infty}$, the
M-V portfolio $b^{\alpha^{*}}$ has either the optimal Sharpe ratio
or the optimal growth rate among the involved M-V portfolios.
\end{cor}
\begin{proof}
By Theorem \ref{Theorem 1}, for any M-V strategy $\big(b_{n}^{\alpha}\big)$
associated with an $\alpha\in\mathcal{A}$, we have $b_{n}^{\alpha}\to b^{\alpha}\in B\big(\alpha,P_{\infty}\big)$,
and subsequently $Sh_{n}\big(b_{n}^{\alpha}\big)\to\mathbb{E}^{P_{\infty}}\big(\big<b^{\alpha},X\big>\big)/\sqrt{\smash[b]{\Var^{P_{\infty}}\big(\big<b^{\alpha},X\big>\big)}}$
and $W_{n}\big(b_{n}^{\alpha}\big)\to\mathbb{E}^{P_{\infty}}\big(\log\big<b^{\alpha},X\big>\big)$,
i.e., $\mathcal{U}\big(b_{n}^{\alpha},P_{n}\big)\to\mathcal{U}\big(b^{\alpha},P_{\infty}\big)$
for the objective function $\mathcal{U}\big(b,Q\big)$ chosen as either
the Sharpe ratio or the expected logarithmic return. The construction
of the M-V strategy $\big(b_{n}^{\alpha_{n}}\big)$ according to (\ref{proposal strategy optimal sharpe})
enables identifying the best M-V strategies among the involved ones.
Specifically, regarding a given objective function $\mathcal{U}\big(b,Q\big)$,
define the set of coefficients with an associated strategy attaining
the optimal limit $\mathcal{U}\big(b^{\alpha^{*}},P_{\infty}\big)$
as follows:
\[
\mathcal{A}^{*}\coloneqq\big\{\alpha\in\mathcal{A}:\,\lim_{n\to\infty}\mathcal{U}\big(b_{n}^{\alpha},P_{n}\big)=\mathcal{U}\big(b^{\alpha},P_{\infty}\big)\text{ such that }\mathcal{U}\big(b^{\alpha},P_{\infty}\big)=\mathcal{U}\big(b^{\alpha^{*}},P_{\infty}\big)\big\}.
\]
Then, let $\Delta\coloneqq\min_{\alpha\notin\mathcal{A}^{*}}|\mathcal{U}\big(b^{\alpha^{*}},P_{\infty}\big)-\mathcal{U}\big(b^{\alpha},P_{\infty}\big)|>0$
denote the minimum distance between the optimal limit and other sub-optimal
ones. Moreover, since there exists $N^{\delta}$ such that $|\mathcal{U}\big(b_{n}^{\alpha},P_{n}\big)-\mathcal{U}\big(b^{\alpha},P_{\infty}\big)|<\delta$
for all $n>N^{\delta}$ and all $\alpha\in\mathcal{A}$, for any $\delta>0$;
thus, there exists $N^{\Delta/2}$ such that $\mathcal{U}\big(b_{n}^{\hat{\alpha}},P_{n}\big)>\mathcal{U}\big(b_{n}^{\bar{\alpha}},P_{n}\big)$
for any $\hat{\alpha}\in\mathcal{A}^{*}$ and $\bar{\alpha}\notin\mathcal{A}^{*}$,
for all $n>N^{\Delta/2}$. Therefore, the sets $\bar{B}\left(P_{n}\right)$
only include the portfolios $b_{n}^{\alpha}$ such that $\alpha\in\mathcal{A}^{*}$
for all $n>N^{\Delta/2}$.\smallskip 

Although the set $\mathcal{A}^{*}$ may contain more than one coefficient,
the M-V strategy $\big(b_{n}^{\alpha_{n}}\big)$, in which each portfolio
$b_{n}^{\alpha_{n}}$ is chosen as one among $b_{n}^{\alpha}$ at
time $n$, must converge to one limiting point $\big(b^{\alpha}\big)$
where $\alpha\in\mathcal{A}^{*}$. In detail, for each $\alpha\in\mathcal{A}^{*}$,
since convergent sequences are also Cauchy sequences, there exists
$N^{\epsilon}$ for arbitrarily small $\epsilon>0$ such that $\|b_{n}^{\alpha}-b_{n+1}^{\alpha}\|<\epsilon$
for all $n>N^{\epsilon}$. Consider any two M-V strategies $\big(b_{n}^{\hat{\alpha}}\big)$
and $\big(b_{n}^{\bar{\alpha}}\big)$ with two respective limiting
portfolios $b^{\hat{\alpha}}$ and $b^{\bar{\alpha}}$ such that $\|b^{\hat{\alpha}}-b^{\bar{\alpha}}\|>0$,
we have the following inequality:
\begin{align*}
 & \|b^{\hat{\alpha}}-b_{n}^{\hat{\alpha}}\|+\|b_{n}^{\hat{\alpha}}-b_{n}^{\bar{\alpha}}\|+\|b_{n}^{\bar{\alpha}}-b^{\bar{\alpha}}\|\geq\|b^{\hat{\alpha}}-b^{\bar{\alpha}}\|,\forall n>h,\\
\Rightarrow & \liminf_{n\to\infty}\|b_{n}^{\hat{\alpha}}-b_{n}^{\bar{\alpha}}\|\geq\|b^{\hat{\alpha}}-b^{\bar{\alpha}}\|>0.
\end{align*}
Hence, since $b_{n}^{\alpha_{n}}\coloneqq\operatorname*{argmin}_{b\in\bar{B}\left(P_{n-1}\right)}\|b-b_{n-2}^{\alpha_{n-2}}\|$
by construction, it cannot switch between two strategies infinitely
often once $\|b_{n}^{\alpha}-b_{n+1}^{\alpha}\|<\epsilon\leq\|b^{\hat{\alpha}}-b^{\bar{\alpha}}\|$
for $\alpha\in\big\{\hat{\alpha},\bar{\alpha}\big\}$ but must keep
choosing only one strategy among $\big(b_{n}^{\hat{\alpha}}\big)$
and $\big(b_{n}^{\bar{\alpha}}\big)$ for all $n\geq\max\big\{ N^{\epsilon},N^{\Delta/2}\big\}+2$.
Otherwise, it is trivial that the M-V strategy $\big(b_{n}^{\alpha_{n}}\big)$
converges to a limiting portfolio if $\|b^{\hat{\alpha}}-b^{\bar{\alpha}}\|=0$.\smallskip

With the convergence of the M-V strategy $\big(b_{n}^{\alpha_{n}}\big)$
to one limiting portfolio of the M-V strategies $\big(b_{n}^{\alpha}\big)$,
denote the limiting point of that M-V strategy as $b^{\alpha^{*}}$
for an $\alpha^{*}\in\mathcal{A}$. Similar to the proof of Theorem
\ref{Theorem 1}, by using Lemma \ref{lemma 2}, we obtain either
one of the following results depending on the manual choice of the
function $\mathcal{U}\big(b,Q\big)$ as defined in (\ref{Function U}):
\[
\lim_{n\to\infty}Sh_{n}\big(b_{n}^{\alpha_{n}}\big)=\frac{\mathbb{E}^{P_{\infty}}\big(\big<b^{\alpha^{*}},X\big>\big)}{\sqrt{\smash[b]{\Var^{P_{\infty}}\big(\big<b^{\alpha^{*}},X\big>\big)}}}\text{ or }\lim_{n\to\infty}W{}_{n}\big(b_{n}^{\alpha_{n}}\big)=\mathbb{E}^{P_{\infty}}\big(\log\big<b^{\alpha^{*}},X\big>\big).\smallskip
\]

We finalize the proof by showing that the constant M-V strategy $\big(b^{\alpha^{*}}\big)$
attains either the optimal expected Sharpe ratio or the optimal expected
logarithmic return among M-V portfolios corresponding to $\alpha\in\mathcal{A}$
with the limiting distribution $P_{\infty}$. Indeed, by the definition
of the sets $\bar{B}\left(P_{n}\right)$ and the convergence $b_{n}^{\alpha_{n}}\to b^{\alpha^{*}}$,
we obtain the needed result as follows:
\begin{align*}
 & \mathcal{U}\big(b_{n}^{\alpha_{n}},P_{n}\big)\geq\mathcal{U}\big(b_{n}^{\alpha},P_{n}\big),\,\forall\alpha\in\mathcal{A},\forall n\geq h,\\
\Rightarrow & \mathcal{U}\big(b^{\alpha^{*}},P_{\infty}\big)=\lim_{n\to\infty}\mathcal{U}\big(b_{n}^{\alpha_{n}},P_{n}\big)\geq\lim_{n\to\infty}\mathcal{U}\big(b_{n}^{\alpha},P_{n}\big)=\mathcal{U}\big(b^{\alpha},P_{\infty}\big),\,\forall b^{\alpha}\in B\left(\alpha,P_{\infty}\right),\forall\alpha\in\mathcal{A}.
\end{align*}
Here, the convergence $\mathcal{U}\big(b_{n}^{\alpha_{n}},P_{n}\big)\to\mathcal{U}\big(b^{\alpha^{*}},P_{\infty}\big)$
for the case where $\mathcal{U}\big(b,Q\big)$ is chosen as the Sharpe
ratio follows from that $\big<b_{n}^{\alpha_{n}},\mu^{P_{n}}\big>\big/\big<b_{n}^{\alpha_{n}},\Sigma^{P_{n}}b_{n}^{\alpha_{n}}\big>^{1/2}\to\big<b^{\alpha^{*}},\mu^{P_{\infty}}\big>\big/\big<b^{\alpha^{*}},\Sigma^{P_{\infty}}b^{\alpha^{*}}\big>^{1/2}$
as $P_{n}\to P_{\infty}$ weakly. Meanwhile, when $\mathcal{U}\big(b,Q\big)$
is chosen as the expected logarithmic return, we demonstrate the convergence
$\mathcal{U}\big(b_{n}^{\alpha_{n}},P_{n}\big)\to\mathcal{U}\big(b^{\alpha^{*}},P_{\infty}\big)$
as follows: since the sequence $x_{1}^{\infty}$ is bounded in $\mathcal{M}$,
define $x^{n}\coloneqq\max_{x\in\mathcal{M}}\big\vert\log\big<b_{n}^{\alpha_{n}},x\big>-\log\big<b^{\alpha^{*}},x\big>\big\vert$
as the value maximizing the distance between the logarithmic returns
of the two portfolios $b_{n}^{\alpha_{n}}$ and $b^{\alpha^{*}}$.
Thus, by the triangle inequality and the convergence $b_{n}^{\alpha_{n}}\to b^{\alpha^{*}}$,
we have:
\begin{align*}
0\leq\lim_{n\to\infty}\big\vert\mathbb{E}^{P_{n}}\big(\log\big<b_{n}^{\alpha_{n}},X\big>-\log\big<b^{\alpha^{*}},X\big>\big)\big\vert & \leq\lim_{n\to\infty}\frac{1}{n}\sum_{i=1}^{n}\big\vert\log\big<b_{n}^{\alpha_{n}},x_{i}\big>-\log\big<b^{\alpha^{*}},x_{i}\big>\big\vert\\
 & \leq\lim_{n\to\infty}\big\vert\log\big<b_{n}^{\alpha_{n}},x^{n}\big>-\log\big<b^{\alpha^{*}},x^{n}\big>\big\vert=0,
\end{align*}
so we obtain the desired limit by the weak convergence $P_{n}\to P_{\infty}$
as follows:
\[
\lim_{n\to\infty}\mathbb{E}^{P_{n}}\big(\log\big<b_{n}^{\alpha_{n}},X\big>\big)=\lim_{n\to\infty}\mathbb{E}^{P_{n}}\big(\log\big<b^{\alpha^{*}},X\big>\big)=\mathbb{E}^{P_{\infty}}\big(\log\big<b^{\alpha^{*}},X\big>\big).
\]
All the required results are obtained, completing the proof.
\end{proof}
\begin{rem*}
If the investor opts to seek the M-V portfolio with the optimal Sharpe
ratio over the simplex $\mathcal{B}^{m}$, the risk aversion set within
$\left[0,a\right]$ should be sufficiently large but not too large,
since $a\to\infty$ leads to overly risk-averse behavior, which would
significantly reduce expected values while increasing the variance
of a portfolio's return. The case of normally distributed assets'
returns, mentioned in the previous section, can be considered as an
example of this property.
\end{rem*}

\section{Universality of strategies in stationary market}

In this section, we discuss the universality of the proposed strategies.
Universality refers to the guaranteed performance of an M-V strategy
in any possible market scenario, assuming that the data is treated
as realizations from a random process, whose infinite-dimensional
distribution is also unknown to the investor, rather than a deterministic
sequence as considered in the previous section. In the general case,
other than the example of an i.i.d. process mentioned in the previous
section that implies a unique limiting distribution, if the realizations
of assets' returns come from a stationary process, Corollary \ref{Corollary 2}
asserts the weak convergence of empirical distributions to a random
limiting one for almost all instances of the infinite sequence of
realizations $x_{1}^{\infty}$. Therefore, it covers a broad range
of cases, including the i.i.d. process and the stationary finite-order
Markov process. It is worth noting that the summability condition,
which is required for the underlying algorithms of the proposed strategies
to work, is broader than stationarity.
\begin{cor}
For any stationary process of assets' returns $\left\{ X_{n}\right\} _{n=1}^{\infty}$
defined on a probability space $\left(\Omega,\mathbb{F},\mathbb{\mathbb{P}}\right)$,
the sequence of empirical distributions $\big\{ P_{n}\big\}_{n=1}^{\infty}$
converges weakly to the (random) limiting distribution $P_{\infty}\big(X\big)\coloneqq\mathbb{\mathbb{P}}\big(X_{1}|\mathcal{I}\big)$,
where $\mathcal{I}\coloneqq\big\{ F\in\mathbb{F}:\,T^{-1}F=F\big\}$,
almost surely. Consequently, Theorem \ref{Theorem 1} and Corollary
\ref{Corollary 1} also hold almost surely.\label{Corollary 2}
\end{cor}
\begin{proof}
Since $\left\{ X_{n}\right\} _{n=1}^{\infty}$ is stationary, by applying
Birkhoff's pointwise ergodic theorem for all bounded and continuous
real functions $f\left(x\right)$, the following limit holds:
\[
\lim_{n\to\infty}\mathbb{E}^{P_{n}}\big(f\big(X\big)\big)=\lim_{n\to\infty}\frac{1}{n}\sum_{i=1}^{n}f\big(X_{i}\big)=\mathbb{E}\big(f\big(X_{1}\big)|\mathcal{I}\big)\eqqcolon\mathbb{E}^{P_{\infty}}\big(f\big(X\big)\big)\text{, a.s,}
\]
thus, $P_{n}\to P_{\infty}$ weakly almost surely. Given the invariant
$\sigma$-field $\mathcal{I}$, the conditional distribution $P_{\infty}$
is the limiting distribution of the sequence of empirical distributions
$\big\{ P_{n}\big\}_{n=1}^{\infty}$, which is correspondingly defined
on each instance of the sequence of realizations $x_{1}^{\infty}$.
The first and second moments of $X$ are finite with respect to the
limiting distribution $P_{\infty}$, so both Theorem \ref{Theorem 1}
and Corollary \ref{Corollary 1} hold almost surely.
\end{proof}

\subsection{Performance comparison with the Bayesian strategy}

In the case of a stochastic process of assets' returns, we are interested
in the long-term behavior of the empirical performance of an M-V strategy
that utilizes the true conditional distribution at any time during
the investment, i.e., the portfolio selection at each period is decided
with a different risk profile. There appears to be no theoretical
research in the literature on M-V portfolio selection addressing this
question, although it is essential, as an investor may argue that
a dynamic M-V strategy can achieve higher empirical M-V utility, Sharpe
ratio, and growth rate when past market information is utilized for
portfolio construction. However, in the next Theorem \ref{Theorem 2},
besides establishing the relevant limiting performances, we demonstrate
a particular case of a stationary market process without redundant
assets, where a dynamic M-V strategy, employing the model parameters
corresponding to the true conditional distribution based on observed
market information, almost surely does not yield higher empirical
M-V utility, Sharpe ratio, and growth rate than certain simple constant
M-V strategies. In other words, the results imply that when the investor
tries to estimate exactly the M-V portfolios with respect to the true
conditional distributions, given the increasing observed information,
there are situations in which such decisions, even when exact estimations
are attainable, are not beneficial but very costly.\smallskip

\textbf{Bayesian M-V strategy with constant risk aversion}. Consider
the process $\left\{ X_{n}\right\} _{n=1}^{\infty}$ defined on a
probability space $\left(\Omega,\mathbb{F},\mathbb{\mathbb{P}}\right)$,
at each time $n\geq2$, the investor observes past realizations of
$X_{1}^{n-1}$ and selects a portfolio using the parameters $\mu^{\mathbb{\mathbb{P}}(X_{n}|\sigma(X_{1}^{n-1}))}$
and $\Sigma^{\mathbb{\mathbb{P}}(X_{n}|\sigma(X_{1}^{n-1}))}$ of
the M-V model, corresponding to the true regular conditional distribution
$\mathbb{\mathbb{P}}\big(X_{n}|\sigma\big(X_{1}^{n-1}\big)\big)$
given $\sigma\big(X_{1}^{n-1}\big)$ (noting that any random variable
defined on a perfect probability space can admit a regular conditional
distribution given any sub-$\sigma$-field, as demonstrated in \citet{Jirina1954}).
Here, the sub-$\sigma$-field $\sigma\big(X_{1}^{n-1}\big)\subseteq\mathbb{F}$
encompasses the past information known to the investor until time
$n$ before selecting a portfolio, which increases as $n\to\infty$
up to a limiting $\sigma$-field $\sigma\big(X_{1}^{\infty}\big)\subseteq\mathbb{F}$
that represents the maximal information available. Given a risk aversion
$\alpha$, let us define the so-called Bayesian M-V strategy, denoted
as $\big(b_{Q_{n}}^{\alpha}\big)$ henceforth, as follows:
\begin{equation}
b_{Q_{1}}^{\alpha}\in B\big(\alpha,Q_{1}\big)\text{ and }b_{Q_{n}}^{\alpha}\in B\big(\alpha,Q_{n}\big),\forall n\geq2,\label{Bayesian strategy}
\end{equation}
where $Q_{1}\coloneqq\mathbb{\mathbb{P}}\big(X_{1}\big)$ is conventionally
taken as the unconditional distribution of the first variable $X_{1}$
without past information, while $Q_{n}\coloneqq\mathbb{\mathbb{P}}\big(X_{n}|\sigma\big(X_{1}^{n-1}\big)\big)$
represents the regular conditional distribution of $X_{n}$ given
past information. Noting that the set $B\big(\alpha,Q_{n}\big)$ is
non-empty as the first two moments of bounded variables $X_{n}$ are
finite for any conditional distribution, as stated in Lemma \ref{Lemma 1}.\smallskip

Before stating the theorem, let us highlight some important properties
of stationarity and time-reversibility. By Kolmogorov\textquoteright s
theorem, a stationary process $\left\{ X_{n}\right\} _{n=1}^{\infty}$
can be embedded into a two-sided stationary process $\left\{ X_{n}\right\} _{-\infty}^{\infty}$
with the origin at $X_{1}$ and an invertible measure-preserving shift
operator, denoted by $T$. Therefore, stationarity implies the equality
of regular conditional distributions $\mathbb{\mathbb{P}}\big(X_{n}|X_{1}^{n-1}=x_{1}^{n-1}\big)=\mathbb{\mathbb{P}}\big(X_{0}|X_{1-n}^{-1}=x_{1-n}^{-1}\big)$
for all $n\geq2$, given the same event of realization sequences $x_{1}^{n-1}=x_{1-n}^{-1}$
(index-wise equality). A stricter variant of a stationary process
is a time-reversible process. Concretely, a process $\left\{ X_{n}\right\} _{n=1}^{\infty}$
is deemed time-reversible if $\mathbb{P}\big(X_{1}^{n}\big)\coloneqq\mathbb{P}\big(X_{1},...,X_{n}\big)=\mathbb{P}\big(X_{n},...,X_{1}\big)=\mathbb{P}\big(X_{n}^{1}\big)$
for all $n\geq2$, i.e., the joint distribution of $X_{1}^{n}$ is
the same as that of its reversed sequence $X_{1}^{n}$. Roughly speaking,
given a common present observation, the conditional distribution of
the future given the present is identical to that of the past given
the present when viewed in reverse time; for instance, any univariate
stationary Gaussian process satisfies time-reversibility. Moreover,
a time-reversible process is always stationary, but the reverse implication
does not generally hold; thus, a time-reversible process can also
be extended into a two-sided stationary process. Consequently, under
the time-reversible process $\left\{ X_{n}\right\} _{n=1}^{\infty}$,
any sequence of variables $X_{k}^{n}$ satisfies $\mathbb{P}\big(X_{k}^{n}\big)=\mathbb{P}\big(X_{n}^{k}\big)=\mathbb{P}\big(X_{-n}^{-k}\big)$
for all $n\geq1$ and $0\leq k<n$, which leads to $\mathbb{\mathbb{P}}\big(X_{0}|\sigma\big(X_{1-n}^{-1}\big)\big)=\mathbb{\mathbb{P}}\big(X_{0}|\sigma\big(X_{1}^{n-1}\big)\big)=\mathbb{\mathbb{P}}\big(X_{n}|\sigma\big(X_{1}^{n-1}\big)\big)$
for all $n\geq2$.\smallskip
\begin{thm}
Consider the market as a process of assets' returns $\left\{ X_{n}\right\} _{n=1}^{\infty}$
defined on the canonical probability space $\left(\Omega,\mathbb{F},\mathbb{\mathbb{P}}\right)$
and its corresponding empirical distributions $\big\{ P_{n}\big\}_{n=1}^{\infty}$.
Assume that the market does not have redundant assets with respect
to the regular conditional distribution $\bar{Q}_{\infty}\coloneqq\mathbb{\mathbb{P}}\big(X_{1}|\sigma\big(X_{-\infty}^{0}\big)\big)$
almost surely, the following properties hold for any risk aversion
$\alpha$:\label{Theorem 2}
\begin{enumerate}
\item If the process $\left\{ X_{n}\right\} _{n=1}^{\infty}$ is stationary,
then the empirical M-V utility of the Bayesian M-V strategy $\big(b_{Q_{n}}^{\alpha}\big)$,
according to (\ref{Bayesian strategy}), converges almost surely to
a random limit as follows:
\begin{align*}
 & \lim_{n\to\infty}\big(M_{n}\big(b_{Q_{n}}^{\alpha}\big)-\alpha V_{n}\big(b_{Q_{n}}^{\alpha}\big)\big)\\
=\, & \mathbb{E}\big(\big<b_{\bar{Q}_{\infty}}^{\alpha},X_{1}\big>|\mathcal{I}\big)-\alpha\mathbb{E}\big(\big<b_{\bar{Q}_{\infty}}^{\alpha},X_{1}\big>-\mathbb{E}\big(\big<b_{\bar{Q}_{\infty}}^{\alpha},X_{1}\big>|\mathcal{I}\big)|\mathcal{I}\big)^{2},\text{ a.s},
\end{align*}
where $\big<b_{\bar{Q}_{\infty}}^{\alpha},X_{1}\big>$ is treated
as a random variable, with $b_{\bar{Q}_{\infty}}^{\alpha}\in B\big(\alpha,\bar{Q}_{\infty}\big)$,
and $\mathcal{I}\coloneqq\big\{ F\in\mathbb{F}:\,T^{-1}F=F\big\}$
denotes the invariant sub-$\sigma$-field of $\mathbb{F}$.\smallskip
\item If the process $\left\{ X_{n}\right\} _{n=1}^{\infty}$ is time-reversible,
then the random limiting empirical M-V utility of the Bayesian M-V
strategy $\big(b_{Q_{n}}^{\alpha}\big)$ as above is almost surely
not greater than the optimal M-V utility of a portfolio corresponding
to the random distribution $P_{\infty}$, as:
\begin{align*}
\lim_{n\to\infty}\big(M_{n}\big(b_{Q_{n}}^{\alpha}\big)-\alpha V_{n}\big(b_{Q_{n}}^{\alpha}\big)\big)\leq\max_{b\in\mathcal{B}^{m}}\mathcal{L}\big(\alpha,b,\mu^{P_{\infty}},\Sigma^{P_{\infty}}\big) & ,\text{ a.s},
\end{align*}
where $P_{\infty}$ denotes the random limit of the sequence of empirical
distributions $\big\{ P_{n}\big\}_{n=1}^{\infty}$.
\end{enumerate}
\end{thm}
\begin{proof}
To prove assertion (a), we first denote the unconditional distribution
as $\bar{Q}_{1}\coloneqq\mathbb{\mathbb{P}}\big(X_{1}\big)=Q_{1}$
by convention, and the regular conditional distribution as $\bar{Q}_{n}\coloneqq\mathbb{\mathbb{P}}\big(X_{1}|\sigma\big(X_{2-n}^{0}\big)\big)$,
which is admitted by $X_{1}$ given the inverse sub-$\sigma$-fields
$\sigma\big(X_{2-n}^{0}\big)\subseteq\mathbb{F}$, for all $n\geq2$.
Then, given the sequence of distributions $\big\{\bar{Q}_{n}\big\}_{n=1}^{\infty}$,
we accordingly define a dynamic M-V strategy $\big(b_{\bar{Q}_{n}}^{\alpha}\big)$
as follows:
\[
b_{\bar{Q}_{1}}^{\alpha}\in B\big(\alpha,\bar{Q}_{1}\big)=B\big(\alpha,Q_{1}\big)\text{ and }b_{\bar{Q}_{n}}^{\alpha}\in B\big(\alpha,\bar{Q}_{n}\big),\forall n\geq2.
\]
Moreover, we can apply the Kuratowski and Ryll-Nardzewski measurable
selection theorem to argue that the M-V portfolio $b_{\bar{Q}_{n}}^{\alpha}$
can be measurably selected depending on $\bar{Q}_{n}$ for any $n$,
since the set $B\big(\alpha,Q\big)$ is non-empty, compact, and convex
for any $Q\in\mathcal{P}$, while $\max_{b\in\mathcal{B}^{m}}\mathcal{L}\big(\alpha,b,\mu^{Q},\Sigma^{Q}\big)$
is a continuous function in $Q\in\mathcal{P}$ by Lemma \ref{Lemma 1},
where $\mathcal{P}$ is compact and metrizable when equipped with
the weak topology. Thus, this argument also implies that the portfolio
$b_{\bar{Q}_{\infty}}^{\alpha}$ can be measurably selected depending
on $\bar{Q}_{\infty}$.\smallskip

Then, given the inverse filtration $\big\{\sigma\big(X_{2-n}^{0}\big)\big\}_{n=2}^{\infty}$,
the weak convergence $\bar{Q}_{n}\to\bar{Q}_{\infty}$ holds almost
surely due to Levy's martingale convergence theorem for conditional
expectation as follows:
\begin{equation}
\lim_{n\to\infty}\mathbb{E}^{\bar{Q}_{n}}\big(f\big(X_{1}\big)\big)=\lim_{n\to\infty}\mathbb{E}\big(f\big(X_{1}\big)|\sigma\big(X_{2-n}^{0}\big)\big)=\mathbb{E}\big(f\big(X_{1}\big)|\sigma\big(X_{-\infty}^{0}\big)\big)\eqqcolon\mathbb{E}^{\bar{Q}_{\infty}}\big(f\big(X_{1}\big)\big),\text{ a.s,}\label{Levy}
\end{equation}
which applies for all bounded and continuous real functions $f(\cdot)$.
Moreover, due to the assumption of no redundant assets almost surely
with respect to the limiting distribution $\bar{Q}_{\infty}$ for
the market, the corresponding covariance matrix $\Sigma^{\bar{Q}_{\infty}}$
is positive definite almost surely, implying that the corresponding
set $B\big(\alpha,\bar{Q}_{\infty}\big)$ is a singleton. Therefore,
by Lemma \ref{Lemma 1}, we obtain the convergence $b_{\bar{Q}_{n}}^{\alpha}\big(\omega\big)\to b_{\bar{Q}_{\infty}}^{\alpha}\big(\omega\big)$
almost surely, with $b_{\bar{Q}_{\infty}}^{\alpha}\in B\big(\alpha,\bar{Q}_{\infty}\big)$,
which immediately results in $\big<b_{\bar{Q}_{n}}^{\alpha},X_{1}\big>\big(\omega\big)\to\big<b_{\bar{Q}_{\infty}}^{\alpha},X_{1}\big>\big(\omega\big)$
almost surely.\smallskip

Since $\sigma\big(X_{1}^{n-1}\big)=T^{(n-1)}\sigma\big(X_{2-n}^{0}\big)$,
the Bayesian M-V strategy $\big(b_{Q_{n}}^{\alpha}\big)\big(\omega\big)$
can be equivalently represented by the shifted portfolios $b_{\bar{Q}_{n}}^{\alpha}\big(\omega\big)$
as $b_{Q_{n}}^{\alpha}\big(\omega\big)=b_{\bar{Q}_{n}}^{\alpha}\big(T^{(n-1)}\omega\big)$
for all $n\geq2$, while $b_{Q_{1}}^{\alpha}\big(\omega\big)=b_{\bar{Q}_{1}}^{\alpha}\big(\omega\big)$
by definition. As a result, we obtain the following limit:
\begin{align*}
\lim_{n\to\infty}M_{n}\big(b_{Q_{n}}^{\alpha}\big)\big(\omega\big)=\lim_{n\to\infty}\frac{1}{n}\sum_{i=1}^{n}\big<b_{Q_{i}}^{\alpha},X_{i}\big>\big(\omega\big) & =\lim_{n\to\infty}\frac{1}{n}\sum_{i=1}^{n}\big<b_{\bar{Q}_{i}}^{\alpha},X_{1}\big>\big(T^{(i-1)}\omega\big)\\
 & =\mathbb{E}\big(\big<b_{\bar{Q}_{\infty}}^{\alpha},X_{1}\big>|\mathcal{I}\big),\text{ a.s},
\end{align*}
by Breiman's generalized ergodic theorem for the $L^{1}$-dominated
sequence $\big\{\big<b_{\bar{Q}_{n}}^{\alpha},X_{n}\big>\big\}_{n=1}^{\infty}$,
i.e., $\mathbb{E}\big(\sup_{n}|\big<b_{\bar{Q}_{n}}^{\alpha},X_{n}\big>|\big)<\infty$
as $|\big<b,X_{n}\big>|$ is bounded for any $n$ and $b\in\mathcal{B}^{m}$.
Subsequently, we have:
\begin{align*}
\lim_{n\to\infty}V_{n}\big(b_{Q_{n}}^{\alpha}\big) & =\lim_{n\to\infty}\frac{1}{n}\sum_{j=1}^{n}\Big(\big<b_{Q_{j}}^{\alpha},X_{j}\big>-\frac{1}{n}\sum_{i=1}^{n}\big<b_{Q_{i}}^{\alpha},X_{i}\big>\Big)^{2}\\
 & =\lim_{n\to\infty}\frac{1}{n}\sum_{j=1}^{n}\Big(\big<b_{\bar{Q}_{j}}^{\alpha},X_{1}\big>\big(T^{(i-1)}\omega\big)-\frac{1}{n}\sum_{i=1}^{n}\big<b_{\bar{Q}_{i}}^{\alpha},X_{1}\big>\big(T^{(i-1)}\omega\big)\Big)^{2}\\
 & =\lim_{n\to\infty}\frac{1}{n}\sum_{j=1}^{n}\big<b_{\bar{Q}_{j}}^{\alpha},X_{1}\big>^{2}\big(T^{(i-1)}\omega\big)-\lim_{n\to\infty}\Big(\frac{1}{n}\sum_{i=1}^{n}\big<b_{\bar{Q}_{i}}^{\alpha},X_{1}\big>\big(T^{(i-1)}\omega\big)\Big)^{2}\\
 & =\mathbb{E}\big(\big<b_{\bar{Q}_{\infty}}^{\alpha},X_{1}\big>^{2}|\mathcal{I}\big)-\mathbb{E}\big(\big<b_{\bar{Q}_{\infty}}^{\alpha},X_{1}\big>|\mathcal{I}\big)^{2},\text{ a.s}.\\
 & =\mathbb{E}\big(\big<b_{\bar{Q}_{\infty}}^{\alpha},X_{1}\big>-\mathbb{E}\big(\big<b_{\bar{Q}_{\infty}}^{\alpha},X_{1}\big>|\mathcal{I}\big)|\mathcal{I}\big)^{2}.
\end{align*}
Thus, the needed limit of the empirical M-V utility of the Bayesian
M-V strategy $\big(b_{Q_{n}}^{\alpha}\big)$ is obtained.\smallskip

To prove assertion (b), noting that since the process is time-reversible,
the sequence of regular conditional distributions $\big\{ Q_{n}\big\}_{n=1}^{\infty}$
can be equivalently represented as follows:
\[
Q_{1}\coloneqq\mathbb{\mathbb{P}}\big(X_{1}\big)=\mathbb{\mathbb{P}}\big(X_{0}\big)\text{ and }Q_{n}\coloneqq\mathbb{\mathbb{P}}\big(X_{n}|\sigma\big(X_{1}^{n-1}\big)\big)=\mathbb{\mathbb{P}}\big(X_{0}|\sigma\big(X_{1}^{n-1}\big)\big),\forall n\geq2.
\]
Hence, by invoking Levy's martingale convergence theorem as in (\ref{Levy}),
we obtain the weak convergence $Q_{n}\to Q_{\infty}$ almost surely,
where $Q_{\infty}\coloneqq\mathbb{\mathbb{P}}\big(X_{0}|\sigma\big(X_{1}^{\infty}\big)\big)$.
In addition, since time-reversibility also implies $\mathbb{\mathbb{P}}\big(X_{0}|\sigma\big(X_{1}^{\infty}\big)\big)=\mathbb{\mathbb{P}}\big(X_{0}|\sigma\big(X_{-\infty}^{-1}\big)\big)$
and $\mathbb{\mathbb{P}}\big(X_{-\infty}^{0}\big)=\mathbb{\mathbb{P}}\big(X_{-\infty}^{1}\big)$,
the corresponding covariance matrix $\Sigma^{Q_{\infty}}$ must be
positive definite almost surely, as the market is assumed not to have
redundant assets with respect to $\bar{Q}_{\infty}$ almost surely.
This results in the singleton set $B\big(\alpha,Q_{\infty}\big)$
and the almost sure convergence $b_{Q_{n}}^{\alpha}\big(\omega\big)\to b_{Q_{\infty}}^{\alpha}\big(\omega\big)$,
where $b_{Q_{\infty}}^{\alpha}\in B\big(\alpha,Q_{\infty}\big)$,
by Lemma \ref{Lemma 1}.\smallskip

Noting that the empirical distributions $P_{n}\to P_{\infty}$ weakly
almost surely, where the random limiting distribution is given by
$P_{\infty}\big(\omega\big)=\mathbb{\mathbb{P}}\big(X_{1}|\mathcal{I}\big)\big(\omega\big)$
due to Corollary \ref{Corollary 2}. Consequently, applying Lemma
\ref{lemma 2} yields that the empirical M-V utility of the M-V strategy
$\big(b_{Q_{n}}^{\alpha}\big)\big(\omega\big)$ converges almost surely
to that of the random constant M-V strategy $\big(b_{Q_{\infty}}^{\alpha}\big)\big(\omega\big)$,
i.e., the limiting portfolio $b_{Q_{\infty}}^{\alpha}\big(\omega\big)$,
which depends on the realization sequence of the process $\left\{ X_{n}\big(\omega\big)\right\} _{n=1}^{\infty}$,
as follows:
\begin{align}
\lim_{n\to\infty}\big(M_{n}\big(b_{Q_{n}}^{\alpha}\big)-\alpha V_{n}\big(b_{Q_{n}}^{\alpha}\big)\big)\big(\omega\big) & =\lim_{n\to\infty}\mathcal{L}\big(\alpha,b_{Q_{\infty}}^{\alpha},\mu^{P_{n}},\Sigma^{P_{n}}\big)\big)\big(\omega\big)\nonumber \\
 & =\mathcal{L}\big(\alpha,b_{Q_{\infty}}^{\alpha}\big(\omega\big),\mu^{P_{\infty}}\big(\omega\big),\Sigma^{P_{\infty}}\big(\omega\big)\big),\text{ a.s},\label{eq: theorem 2 proof}\\
 & \leq\max_{b\in\mathcal{B}^{m}}\mathcal{L}\big(\alpha,b,\mu^{P_{\infty}}\big(\omega\big),\Sigma^{P_{\infty}}\big(\omega\big)\big).\nonumber 
\end{align}
Accordingly, the dynamic Bayesian M-V strategy $\big(b_{Q_{n}}^{\alpha}\big)$
and the constant one $\big(b_{Q_{\infty}}^{\alpha}\big)$ also yield
the same limiting empirical Sharpe ratio and growth rate. All the
required properties are finally obtained, completing the proof.
\end{proof}
\begin{rem*}
The construction (\ref{Bayesian strategy}) of the Bayesian M-V strategy
can be extended to incorporate an adapted risk aversion coefficient
$\alpha_{n}=\alpha\big(Q_{n}\big)$, corresponding to the true conditional
distribution at each period, rather than a fixed $\alpha$. This implies
that each M-V model associated with a conditional distribution has
an appropriately chosen risk aversion. In this extended context, if
the risk aversion can be constructed as a function $\alpha\big(Q\big)$
that is continuous in $Q\in\mathcal{P}$, then both statements of
Theorem \ref{Theorem 2} remain valid.
\end{rem*}
As a consequence of Theorem \ref{Theorem 2}, for a given risk aversion,
when the time-reversible market process causes the Bayesian M-V strategy
$\big(b_{Q_{n}}^{\alpha}\big)$ to almost surely yield a suboptimal
M-V utility of a portfolio corresponding to the limiting empirical
distribution $P_{\infty}$, it may also result in a lower empirical
Sharpe ratio than that of the corresponding optimal M-V portfolio
located on the efficient frontier. Moreover, if the limiting empirical
distribution $P_{\infty}$ is normal, the asymptotic growth rate of
the Bayesian M-V strategy $\big(b_{Q_{n}}^{\alpha}\big)$ almost surely
tends to increase with a higher risk aversion coefficient $\alpha$,
up to a certain point, as argued by Proposition \ref{Normal mean variance tradeoff}.
Hence, under a time-reversible market process with some additional
conditions, the next Corollary \ref{Corollary 3} guarantees that
our proposed algorithms in previous sections, which do not require
estimations of the conditional distributions based on past observations
at any time, can construct an M-V strategy that almost surely yields
the empirical M-V utility and either the empirical Sharpe ratio or
growth rate, not lower than those yielded by the Bayesian M-V strategy.
\begin{cor}
Consider the market as a time-reversibility process as mentioned in
Theorem \ref{Theorem 2}, then, two following properties hold for
a given risk aversion $\alpha$:\label{Corollary 3} 
\begin{enumerate}
\item The limiting empirical M-V utility of the M-V strategy $\big(b_{n}^{\alpha}\big)$,
according to (\ref{Proposal strategy mean variance}), is almost surely
not lower than that of the Bayesian M-V strategy $\big(b_{Q_{n}}^{\alpha}\big)$
for any risk aversion $\alpha$ as:
\[
\lim_{n\to\infty}\big(M_{n}\big(b_{Q_{n}}^{\alpha}\big)-\alpha V_{n}\big(b_{Q_{n}}^{\alpha}\big)\big)\leq\lim_{n\to\infty}\big(M_{n}\big(b_{n}^{\alpha}\big)-\alpha V_{n}\big(b_{n}^{\alpha}\big)\big),\text{ a.s}.
\]
\item Consider a set of risk aversion $\mathcal{A}$, including the coefficient
$\alpha^{*}$, corresponding to the M-V portfolio $b^{\alpha^{*}}\in B\big(\alpha^{*},P_{\infty}\big)$,
which yields either the highest Sharpe ratio or the highest expected
logarithmic return with respect to $P_{\infty}$ among all $b\in\mathcal{B}^{m}$.
Depending on the choice of the objective function $\mathcal{U}\big(b,Q\big)$,
as defined in (\ref{Function U}), either the limiting empirical Sharpe
ratio or the limiting growth rate yielded by the M-V strategy $\big(b_{n}^{\alpha_{n}}\big)$,
constructed according to (\ref{proposal strategy optimal sharpe}),
is almost surely not lower than that of the Bayesian M-V strategy
$\big(b_{Q_{n}}^{\alpha}\big)$ as:
\[
\lim_{n\to\infty}Sh_{n}\big(b_{n}^{\alpha_{n}}\big)\geq\lim_{n\to\infty}Sh_{n}\big(b_{Q_{n}}^{\alpha}\big)\text{ or }\lim_{n\to\infty}W{}_{n}\big(b_{n}^{\alpha_{n}}\big)\geq\lim_{n\to\infty}W_{n}\big(b_{Q_{n}}^{\alpha}\big),\forall\alpha\in\mathcal{A},\text{ a.s.}
\]
\end{enumerate}
\end{cor}
\begin{proof}
Recalling that the assumption of non-redundant assets for assertion
(b) of Theorem \ref{Theorem 2} implies that the covariance matrix
$\Sigma^{Q_{\infty}}$ corresponding to the limiting conditional distribution
$Q_{\infty}\coloneqq\mathbb{\mathbb{P}}\big(X_{0}|\sigma\big(X_{1}^{\infty}\big)\big)$
is positive definite, as $Q_{\infty}$ has full-dimensional support,
almost surely. Thus, by taking the expectation of $Q_{\infty}$ conditioned
on the $\sigma$-field $\mathcal{I}\subseteq\sigma\big(X_{1}^{\infty}\big)$,
we have $\mathbb{E}\big(\mathbb{\mathbb{P}}\big(X_{0}|\sigma\big(X_{1}^{\infty}\big)\big)|\mathcal{I}\big)=\mathbb{\mathbb{P}}\big(X_{0}|\mathcal{I}\big)=\mathbb{\mathbb{P}}\big(X_{1}|\mathcal{I}\big)$
almost surely, which implies that the support of the limiting empirical
distribution $P_{\infty}=\mathbb{\mathbb{P}}\big(X_{1}|\mathcal{I}\big)$,
according to Corollary \ref{Corollary 2}, is full-dimensional. Consequently,
the corresponding covariance matrix $\Sigma^{P_{\infty}}$ is positive
definite almost surely.\smallskip 

To show assertion (a), since the set $B\big(\alpha,P_{\infty}\big)$
associated with $\Sigma^{P_{\infty}}$ is a singleton almost surely,
Theorem \ref{Theorem 1} ensures that the M-V strategy $\big(b_{n}^{\alpha}\big)$,
formed according to (\ref{Proposal strategy mean variance}), has
an empirical M-V utility asymptotically approaching the utility of
the M-V portfolio $b^{\alpha}\in B\big(\alpha,P_{\infty}\big)$. Consequently,
the following inequality holds almost surely by applying Theorem \ref{Theorem 2}
to a time-reversible process:
\begin{align*}
\lim_{n\to\infty}\big(M_{n}\big(b_{n}^{\alpha}\big)-\alpha V_{n}\big(b_{n}^{\alpha}\big)\big) & =\max_{b\in\mathcal{B}^{m}}\mathcal{L}\big(\alpha,b,\mu^{P_{\infty}},\Sigma^{P_{\infty}}\big)\\
 & \geq\lim_{n\to\infty}\big(M_{n}\big(b_{Q_{n}}^{\alpha}\big)-\alpha V_{n}\big(b_{Q_{n}}^{\alpha}\big)\big),\text{ a.s.}
\end{align*}
Thus, assertion (a) is confirmed; however, it is important to stress
that the above inequality does not exclude the possibility that the
Bayesian M-V strategy $\big(b_{Q_{n}}^{\alpha}\big)$ could yield
a higher empirical Sharpe ratio or growth rate than that yielded by
the M-V strategy $\big(b_{n}^{\alpha}\big)$.\smallskip 

The proof for assertion (b) follows as a consequence. Specifically,
we use the convergence of the sequences $\big\{ M_{n}\big(b_{Q_{n}}^{\alpha}\big)\big\}_{n=1}^{\infty}$
and $\big\{ V_{n}\big(b_{Q_{n}}^{\alpha}\big)\big\}_{n=1}^{\infty}$
as established in (\ref{eq: theorem 2 proof}), which follows from
the almost sure convergence $b_{Q_{n}}^{\alpha}\to b_{Q_{\infty}}^{\alpha}$.
Furthermore, the set $\mathcal{A}$ is assumed to include the coefficient
$\alpha^{*}$, with $b^{\alpha^{*}}$ yielding either the highest
Sharpe ratio or the highest expected logarithmic return with respect
to $P_{\infty}$ among all portfolios $b\in\mathcal{B}^{m}$. Therefore,
the M-V strategy $\big(b_{n}^{\alpha_{n}}\big)$, constructed according
to (\ref{proposal strategy optimal sharpe}), results in the inequality
below if the function $\mathcal{U}\big(b,Q\big)$, as defined in (\ref{Function U}),
is chosen as the expected logarithmic return:
\[
\lim_{n\to\infty}\big(W{}_{n}\big(b_{n}^{\alpha_{n}}\big)-W_{n}\big(b_{Q_{n}}^{\alpha}\big)\big)=\mathbb{E}^{P_{\infty}}\big(\log\big<b^{\alpha^{*}},X\big>\big)-\mathbb{E}^{P_{\infty}}\big(\log\big<b_{Q_{\infty}}^{\alpha},X\big>\big)\geq0,\forall\alpha\in\mathcal{A},\text{ a.s.}
\]
Instead, if the function $\mathcal{U}\big(b,Q\big)$ is chosen as
the Sharpe ratio, the following inequality holds:
\[
\lim_{n\to\infty}\big(Sh_{n}\big(b_{n}^{\alpha_{n}}\big)-Sh_{n}\big(b_{Q_{n}}^{\alpha}\big)\big)=\frac{\mathbb{E}^{P_{\infty}}\big(\big<b^{\alpha^{*}},X\big>\big)}{\sqrt{\smash[b]{\Var^{P_{\infty}}\big(\big<b^{\alpha^{*}},X\big>\big)}}}-\frac{\mathbb{E}^{P_{\infty}}\big(\big<b_{Q_{\infty}}^{\alpha},X\big>\big)}{\sqrt{\smash[b]{\Var^{P_{\infty}}\big(\big<b_{Q_{\infty}}^{\alpha},X\big>\big)}}}\geq0,\forall\alpha\in\mathcal{A},\text{ a.s.}\vspace{0.7ex}
\]
Hence, we obtain all the needed results and complete the proof.
\end{proof}
\begin{rem*}
To invoke assertion (b) of Corollary \ref{Corollary 3}, the set $\mathcal{A}$
must include the risk aversion $\alpha^{*}$ as required by the conditions
of the statement. Clearly, the coefficient $\alpha^{*}$ is unknown
to the investor since the random limiting distribution $P_{\infty}$
is unknown. However, assume that the true $\alpha^{*}$ is not too
large; then, by considering $\mathcal{A}$ as a discretization of
an interval $\left[0,a\right]$ that can include the theoretical $\alpha^{*}$
with a sufficiently large $a$, the resulting M-V strategy $\big(b_{Q_{n}}^{\alpha}\big)$
can attain near-optimal limiting performance compared to that yielded
by the true constant M-V strategy $\big(b^{\alpha^{*}}\big)$ as stated.
\end{rem*}

\section{Concluding remarks and potential extensions}

In this paper, we propose an online learning approach to address the
long-standing problem of the Markowitz optimization enigma, which
describes the challenge posed by unknown parameters corresponding
to a distribution in the M-V model. The algorithms for constructing
the dynamic strategy guarantee asymptotically approaching the same
empirical M-V utility and other empirical performances, including
the Sharpe ratio and growth rate, as the constant M-V strategy derived
from the M-V model with knowledge of the distribution of assets' returns.
The proposed strategies work in markets without redundant assets under
a minimal summability condition for the infinite data sequence, which
is very general and thus covers a broad range of stochastic markets,
including i.i.d. and stationary processes. Moreover, if the stochastic
market is time-reversible and does not have redundant assets with
respect to the conditional distribution of assets' returns given the
infinite past, the Bayesian M-V strategy using the true conditional
distribution, based on past observations, does not yield a higher
limiting empirical M-V utility, Sharpe ratio, or growth rate than
the proposed strategies, almost surely. Additionally, the relationship
between the Sharpe ratio, expected return, and growth rate through
risk aversion calibration is demonstrated in a market with normally
distributed assets' returns, which also explains the behavior of the
growth rate of the proposed strategies in this specific type of market.\smallskip

We remark on some potential model extensions for further theoretical
and empirical research. The M-V model could be considered with another
risk measure rather than the variance of the portfolio's return, as
it is argued that variance may not be an appropriate risk measure
for heavy-tailed distributions, while also considering the volatility
of increasing returns as risky. For example, the M-V utility can be
modified with higher-order moments of portfolio returns other than
their variance or by considering the variance of the difference in
returns with a market portfolio; in these cases, both the analysis
and proposed strategies in this paper are expected to remain valid.
Besides, it is worth noting that the no-short sell restriction may
be removed, albeit with the trivial caveat that capital can be depleted
to zero in sequential investments. Alternatively, in accordance with
the common framework of online convex optimization, in the general
case of deterministic data with the removed summability condition,
some theoretical computer scientists may find it interesting to bound
the difference of the empirical M-V utility with the best constant
M-V strategy over time to achieve Hannan consistency as in \citet{Hannan1957}.
As mentioned in the paper, the M-V utility provides an alternative
loss function to the mainstream logarithmic return, and it likely
does not lend itself trivially to linear approximation as a common
approach. This problem becomes more challenging if it is required
to simultaneously achieve multi-objectives, including optimal empirical
Sharpe ratio and growth rate as in the paper.

\pagebreak

\end{document}